\documentclass[envcountsame]{llncs}

\usepackage{amsmath}
\usepackage{subfigure}
\usepackage{booktabs}
\usepackage{tikz}
\usetikzlibrary{calc}
\usetikzlibrary{positioning}
\usetikzlibrary{automata}
\usetikzlibrary{shapes}

\usepackage{multirow}

\title{On the Structure and Complexity of \RegularlyGeneratedLanguageSetsTitle}
\author{Andreas Holzer\inst{1} \and Christian Schallhart\inst{2} \and Michael
  Tautschnig\inst{3} \and Helmut~Veith\inst{1}}
\institute{Vienna University of Technology, Austria \and University of Oxford, UK \and Queen Mary, University of London, UK}

\usepackage{epsfig}
\usepackage{amssymb}
\usepackage{xspace}
\usepackage{url}
\usepackage[ruled,vlined,linesnumbered,shortend]{algorithm2e}
\SetKwInOut{Input}{input}
\SetKwInOut{Requires}{requires}
\SetKwInOut{Returns}{returns}
\SetKw{Yield}{yield}
\SetKw{WriteAs}{write}
\SetKw{WriteAsAs}{as}
\SetKw{Violates}{violates}
\SetKw{With}{with}
\SetKwInOut{Yields}{yields}
\usepackage{tabularx}
\usepackage{todonotes}
\usepackage{paralist}

\newcommand{\PSPACE}{\complexityclass{PSpace}}
\newcommand{\DSPACE}{\complexityclass{DSpace}}

\newcommand{\tEXPSPACE}{\complexityclass{2ExpSpace}}
\newcommand{\EXPSPACE}{\complexityclass{ExpSpace}}

\newcommand{\complexityclass}[1]{\ensuremath{\textsc{#1}}\xspace}

\newcommand{\rationalset}{\ensuremath{\mathcal{R}}\xspace}
\newcommand{\representation}{\ensuremath{{\sf rep}}\xspace}
\newcommand{\basiccheck}{\ensuremath{{\sf basiccheck}}\xspace}
\newcommand{\membership}{\ensuremath{{\sf membership}}\xspace}
\newcommand{\setenumeration}{\ensuremath{{\sf enumerate}}\xspace}
\newcommand{\unfold}{\ensuremath{{\sf unfold}}\xspace}
\newcommand{\order}{\ensuremath{{\mathcal O}}\xspace}
\newcommand{\unfoldstep}{\ensuremath{{\sf ufs}}\xspace}

\newcommand{\minlength}{\ensuremath{{\sf minlen}}\xspace}

\newcommand{\cR}{\ensuremath{\mathcal{R}}\xspace}

\newcommand{\bN}{\ensuremath{\mathbb{N}}\xspace}
\newcommand{\cA}{\ensuremath{\mathcal{A}}\xspace}
\newcommand{\cM}{\ensuremath{\mathcal{M}}\xspace}
\newcommand{\boolTrue}{\ensuremath{\mathbf{true}}\xspace}
\newcommand{\boolFalse}{\ensuremath{\mathbf{false}}\xspace}

\newcommand{\dotdiv}{\:\ensuremath{\mathaccent\cdot-}\:}
\newcommand{\dotcup}{\:\ensuremath{\mathaccent\cdot\cup}\:}
\newcommand{\dotcap}{\:\ensuremath{\mathaccent\cdot\cap}\:}
\newcommand{\timescup}{\mathrel{\ooalign{\hss$\scriptscriptstyle\mathaccent\times{ }$\hss\cr$\cup$}}}
\newcommand{\timescap}{\mathrel{\ooalign{\hss$\scriptscriptstyle\mathaccent\times{ }$\hss\cr$\cap$}}}
\newcommand{\timesminus}{\:\ensuremath{\scriptscriptstyle\mathaccent\times{\textstyle-}}\:}

\renewcommand{\epsilon}{\pleaseusevarepsiloninstead}

\newcommand{\unionfreerep}{\ensuremath{\mathsf{unionfreedecomp}}\xspace}
\newcommand{\criticalpos}{\ensuremath{\mathsf{critical}}\xspace}
\newcommand{\Linit}{\ensuremath{L_\mathsf{init}}\xspace}

\newcommand{\query}[1]{\medskip \noindent{\small\hspace*{1em}\tt #1}\medskip\noindent}
\newdimen{\queryindent}
\newcolumntype{C}{>{\centering}X}
\newcolumntype{Z}{>{\centering\arraybackslash}X}

\newcommand{\RegularlyGeneratedLanguageSets}{rational sets of regular languages\xspace}
\newcommand{\RegularlyGeneratedLanguageSet}{rational set of regular languages\xspace}
\newcommand{\RegularlyGeneratedLanguageSetAbbrev}{RSRL\xspace}
\newcommand{\RegularlyGeneratedLanguageSetsAbbrev}{RSRLs\xspace}

\newcommand{\RegularlyGeneratedLanguageSetsStartAbbrev}{RSRL\xspace}
\newcommand{\RegularlyGeneratedLanguageSetsTitle}{Rational Sets of Regular Languages\xspace}
\newcommand{\RegularlyGeneratedLanguageSetsTitleAbbrev}{RSRL\xspace}
\newcommand{\ARegularlyGeneratedLanguageSet}{rational\xspace}
\newcommand{\ARegularlyGeneratedLanguageSetTitle}{Rational\xspace}

\newtheorem{fact}[theorem]{Fact}

\newcommand{\FSHELL}{\textsc{Fshell}\xspace}
\newcommand{\FQL}{FQL\xspace}

\begin{document}

\maketitle

\thispagestyle{plain} 
\pagestyle{plain}

\begin{abstract}
In a recent thread of papers, we have introduced \FQL, a precise specification
language for test coverage, and developed the test case generation engine \FSHELL
for ANSI C. In essence, an \FQL test specification amounts to a set of regular languages,
each of which has to be matched by at least one test execution. To describe such sets of regular
languages, the \FQL semantics uses an automata-theoretic concept
known as \RegularlyGeneratedLanguageSets (\RegularlyGeneratedLanguageSetsAbbrev). \RegularlyGeneratedLanguageSetsAbbrev are automata whose alphabet
consists of regular expressions. Thus, the language accepted by the automaton is a set of
regular expressions.

In this paper, we study RSRLs from a theoretic point of view. More specifically, we analyze
RSRL closure properties under common set theoretic operations, and the complexity of membership checking,
i.e., whether a regular language is an element of a \RegularlyGeneratedLanguageSetAbbrev. For all questions we investigate both the general
case and the case of finite sets of regular languages. Although a few properties are left as open problems,
the paper provides a systematic semantic foundation for the test specification language \FQL.
\end{abstract}

\section{Introduction}
\label{sec:introduction}

Despite the success of model checking and theorem proving, software testing has a dominant role
in industrial practice. In fact, state-of-the-art development guidelines such as the avionic
standard DO-178B~\cite{do-178b} are heavily dependent on test coverage criteria. It is therefore quite
surprising that the formal specification of coverage criteria has been a blind spot in the
formal methods and software engineering communities for a long time.

In a recent thread of
papers~\cite{FQL-ASE,holzer11:_seaml_testin_for_model_and_code,holzer10:_introd_to_test_specif_in_fql,holzer09:_query_dirven_progr_testin,holzer08,DBLP:conf/esop/BeyerHTV13},
we have addressed this situation and introduced the \FSHELL Query
Language (\FQL) to specify and tailor coverage criteria, together with
\FSHELL, a tool to generate matching test suites for ANSI C programs.
At the semantic core of \FQL, test goals are described as regular expressions whose alphabet are the edges
of the program control flow graph (CFG). For example, to cover a particular CFG edge $c$, one can use the
regular expression $\Sigma^\star\ c\ \Sigma^\star$. Importantly, however, a coverage criterion usually contains
not just a single test goal, but a (possibly large) number of test goals -- e.g.~{\em all} basic blocks
of a program. \FQL therefore employs regular languages which can express sets of regular expressions.
To this end, the alphabet contains not only the CFG edges but also \emph{postponed regular expressions} over these edges, written within quotes.

For example, $"\Sigma^\star"\ (a + b+ c +d)\ "\Sigma^\star"$ describes the language $\{"\Sigma^\star"\ a\ "\Sigma^\star",$  $"\Sigma^\star"\ b\ "\Sigma^\star", "\Sigma^\star"\ c\ "\Sigma^\star", "\Sigma^\star"\ d\ "\Sigma^\star"\}$. Each of these words is a regular expression that will then serve as a test goal. Following~\cite{membership},
  we call such languages \emph{\RegularlyGeneratedLanguageSets
  (\RegularlyGeneratedLanguageSetAbbrev).}

The goal of this paper is to initiate a systematic study of \RegularlyGeneratedLanguageSetsAbbrev from a
theoretical point of view, considering closure properties and complexity of common set-theoretic operations.
Thus, this paper is a first step towards a systematic foundation of \FQL. \RegularlyGeneratedLanguageSetsAbbrev have
a similar role for test specifications as relational algebra has for databases. In particular, a good understanding
of set-theoretic operations is necessary for systematic algorithmic optimization and manipulation of test specifications.
First results on query optimization for \FQL have been obtained in \cite{DBLP:conf/esop/BeyerHTV13}.

A \RegularlyGeneratedLanguageSet is given by a regular
language $K$ over alphabet $\Delta$, and a \emph{regular language
  substitution} $\varphi : \Delta \rightarrow 2^{\Sigma^*}$, mapping
each symbol $\delta\in \Delta$ to a regular language $\varphi(\delta)$
over alphabet $\Sigma$.
We extend $\varphi$ to words $w \in \Delta^+$ with $\varphi(\delta
\cdot w) = \varphi(\delta) \cdot \varphi(w)$, and set $\varphi(L) =
\bigcup_{w \in L}\varphi(w)$ for~$L \subseteq \Delta^+$.
%


\begin{definition}[\RegularlyGeneratedLanguageSetsTitle,
  \RegularlyGeneratedLanguageSetsAbbrev~\cite{membership}]
  \label{def:rational-set}
  A set \rationalset of regular languages over $\Sigma$ is called
  \emph{\ARegularlyGeneratedLanguageSet,} written $\rationalset=(K,\varphi)$, if there exists
  a finite alphabet $\Delta$, a regular language $K \subseteq
  \Delta^+$, and a regular language substitution $\varphi : \Delta^+
  \rightarrow 2^{\Sigma^*}$, such that $\rationalset = \{ \varphi(w)
  \mid w \in K \}.$
  The \RegularlyGeneratedLanguageSetAbbrev~\rationalset is \emph{Kleene star free,} if $K$ is
  given as Kleene star free regular expression.
\end{definition}

Depending on context, we refer to \rationalset as a set of languages
or as a pair~$(K,\varphi)$, but we always write $L\in\rationalset$ iff
$\exists w\in K : L=\varphi(w)$.
  \label{ex:intro}
  Consider the above specification $"\Sigma^\star"\ (a + b+ c +d)\
  "\Sigma^\star"$ over base alphabet $\Sigma=\{a,b,c,d\}$.
  To represent this specification as
  \RegularlyGeneratedLanguageSetAbbrev $\rationalset=(K,\varphi)$, we
  set $\Delta=\{\delta_{\Sigma^\star}\}\cup\Sigma$, containing a fresh
  symbol $\delta_{\Sigma^\star}$ for the quoted expression
  $"\Sigma^\star"$.
  We set $K=L(\delta_{\Sigma^\star}\ (a+b+c+d)\
  \delta_{\Sigma^\star})$ with
  $\varphi(\delta_{\Sigma^\star})=\Sigma^\star$ and
  $\varphi(\sigma)=\sigma$ for $\sigma\in\Sigma$.
  Thus $K$ contains the words $\delta_{\Sigma^\star}\ a\
  \delta_{\Sigma^\star},\dots$ with $\varphi(\delta_{\Sigma^\star}\ a\
  \delta_{\Sigma^\star})=L(\Sigma^\star\ a\
  \Sigma^\star)\in\rationalset$, as desired.


%
Note that the \RegularlyGeneratedLanguageSetAbbrev above is finite
with exactly four elements.
This is of course not atypical: in concrete testing applications, \FQL
generates finite sets of test goals, since it relies on \emph{Kleene
  star free} \RegularlyGeneratedLanguageSetsAbbrev only.
For future applications, however, it is well possible to consider
infinite sets of test goals e.g. for unbounded integer and real valued
variables or for path coverage criteria which are either matched
partially, or by abstract executions.
%
%
In this paper, we are therefore considering the general, finite, and
Kleene star free case.

\begin{example}
  Consider the alphabets~$\Delta = \{ \delta_1, \delta_2 \}$
  and~$\Sigma = \{ a, b \}$. Then,
  \begin{inparaenum}[\bfseries(1)]
  \item  with $\varphi(\delta_1) = L(a^\star)$,
    $\varphi(\delta_2) = \{ ab \}$, and $K =
    L(\delta_1\delta_2^\star\delta_1)$, we obtain the
    \RegularlyGeneratedLanguageSet $\{ L(a^\star(ab)^ia^\star) \mid i
    \in \bN \}$;
  \item with $\varphi(\delta_1) = \{ a^i \mid i \geq 0 \}$,
    $\varphi(\delta_2) = \{ a \}$, and $K = \{ \delta_1\delta_2^i \mid
    i \geq 0 \}$, we obtain $\varphi(w_1) \supset \varphi(w_2)$ for
    all $w_1, w_2 \in K$ with $|w_1| < |w_2|$;
  \item with $\varphi(\delta_1) = \{ \varepsilon, a \}$,
    $\varphi(\delta_2) = \{ aa \}$, and $K = \{ \delta_1\delta_2^i
    \mid i \geq 0 \}$, we have $|\varphi(w)| = 2$ and $\varphi(w) \cap
    \varphi(w') = \emptyset$ for all $w\neq w' \in K$.
  \end{inparaenum}
\end{example}

In the finite case we make an additional distinction for the
subcase where the regular expressions in $\Delta$, i.e., the set of
postponed regular expressions, are fixed. This has practical
relevance, because in the context of \FQL, the results of the
operations on \RegularlyGeneratedLanguageSetAbbrev will be better
readable by engineers if $\Delta$ is unchanged.

\paragraph{Contributions and Organization.}
%
In Section~\ref{sec:closure-properties}, we investigate \emph{closure
  properties} of general and finite
\RegularlyGeneratedLanguageSetsAbbrev, considering the operators
product, Kleene star, complement, union, intersection, set difference,
and symmetric difference.
We also consider the case of finite
\RegularlyGeneratedLanguageSetsAbbrev with a fixed language
substitution $\varphi$, as this case is of particular interest for
testing applications.
Next, in Section~\ref{sec:decision-problems}, we investigate the
\emph{complexity} of the decision problems equivalence, inclusion, and
membership for Kleene star free \RegularlyGeneratedLanguageSetsAbbrev.
We also give an algorithm for checking the membership in general and
analyze its complexity.
We close in Section~\ref{sec:conclusion} with a discussion on how our
results reflect back to design decisions for \FQL.

\section{Related Work}

Afonin et al.~\cite{membership} introduced
\RegularlyGeneratedLanguageSetsStartAbbrev and studied the
decidability of whether a regular language is contained in an
\RegularlyGeneratedLanguageSetsStartAbbrev and the decidability of
whether an \RegularlyGeneratedLanguageSetsStartAbbrev is finite.
Although Afonin et al.~shortly discuss possible upper bounds for the
membership decision problem, their analysis is incomplete due to gaps
in their algorithmic presentation. 
%
%
%
%
Pin introduced the term \emph{extended automata} for
\RegularlyGeneratedLanguageSetsAbbrev as an example for a formalism that can be expressed by equations~\cite{Pin}, but
did not investigate any of their properties.
In our own related work on \FQL~\cite{FQL-ASE,holzer12:_provin_reach_using_fshel_compet_contr,holzer11:_seaml_testin_for_model_and_code,holzer10:_introd_to_test_specif_in_fql,holzer09:_query_dirven_progr_testin,holzer08,DBLP:conf/esop/BeyerHTV13}, we deal with practical
issues arising in testcase generation. Note that \FQL uses a language layer on top of 
\RegularlyGeneratedLanguageSetsAbbrev which extracts the alphabet from the program using
a convenient syntax. 
In conclusion, we are unaware of related work that considers the properties we study here.

Let us finally discuss other work whose terminology is similar to \RegularlyGeneratedLanguageSetsAbbrev without
direct technical relation. Barcel{\'o} et al.~define \emph{rational relations}, which are
relations between words over a common alphabet, whereas we consider
sets of regular languages~\cite{DBLP:conf/lics/BarceloFL12}.
Barcel{\'o} et al.~also investigate \emph{parameterized regular
  languages}~\cite{DBLP:journals/tcs/BarceloRL13}, where words are
obtained by replacing variables in expressions with alphabet symbols.
\emph{Metaregular languages} deal with languages recognized by
automata with a time-variant
structure~\cite{1967Agasandyan,Salomaa196885}.
Lattice Automata~\cite{lattice-automata} only consider lattices that
have a unique complement element, whereas
\RegularlyGeneratedLanguageSetsAbbrev are not closed under complement
(no \RegularlyGeneratedLanguageSetAbbrev has a
\RegularlyGeneratedLanguageSetAbbrev as complement).
%
%
%

\section{Closure Properties}
\label{sec:closure-properties}

\paragraph*{Operators.}
\label{sec:operators}
We investigate the closure properties of \RegularlyGeneratedLanguageSetsAbbrev, considering
standard set theoretic operators, such as union, intersection, and
complement, and variants thereof, fitting \RegularlyGeneratedLanguageSetsAbbrev. 
In particular, we apply those operators also to pairs in the
\emph{Cartesian product} of \RegularlyGeneratedLanguageSetsAbbrev, and \emph{point-wise} to
each element in a \RegularlyGeneratedLanguageSetAbbrev and another given regular language.

\begin{definition}[Operations on \RegularlyGeneratedLanguageSetsTitleAbbrev]
  \label{def:ops}
  Let $\rationalset_1$ and $\rationalset_2$ be \RegularlyGeneratedLanguageSetsAbbrev and let
  $R$ be a regular language.  Then, we define the following operations
  on \RegularlyGeneratedLanguageSetsAbbrev:
  \begin{center}
  {
  	\relsize{-1}
    \begin{tabularx}{\textwidth}{p{10.5em}lp{2em}X}
      \toprule
      Operation & Definition & \\
      \midrule
      
      Product & $\rationalset_1 \cdot \rationalset_2$ & = & 
      $\{ L_1 \cdot L_2 \mid L_1 \in \rationalset_1, L_2 \in \rationalset_2 \}$\\
      
      Kleene Star & $\rationalset_1^\star$ & = & 
      $\bigcup_{i \in \bN}\rationalset_1^i$\\
      
      \hspace*{2em}Point-wise & $\dot{\rationalset_1^\star}$ & = & 
      $\{ L^\star \mid L \in \rationalset_1 \}$ \\
      
      Complement & $\overline{\rationalset_1}$ & = & 
      $\{ L \subseteq 2^{\Sigma^*} \mid L \notin \rationalset_1 \}$ \\
      
      \hspace*{2em}Point-wise & $\dot{\overline{\rationalset_1}}$ & = & 
      $\{ \overline{L} \mid L \in \rationalset_1 \}$ \\
      
      Binary Operators & \multicolumn{3}{l}{
        $\rationalset_1\cap\rationalset_2$,\ \ 
        $\rationalset_1\cup\rationalset_2$,\ \ 
        $\rationalset_1-\rationalset_2$\ \ 
        \hfill (standard def.)
      }\\
      
 
     \hspace*{2em}Point-wise\hfill & $\rationalset_1 \dotcup / \dotcap / \dotdiv R$ & = & $\{ L
      \cup / \cap / - R \mid L \in \rationalset_1 \}$ 
      \\
      
      \hspace*{2em}Cartesian & $\rationalset_1 \timescup / \timescap / \timesminus \rationalset_2$ & = & 
      $\{ L_1 \cup / \cap / - L_2 \mid L_1\in \rationalset_1, L_2 \in \rationalset_2 \}$\\
      
      Symmetric Difference & $\rationalset_1 \Delta \rationalset_2$ & = & 
      $\{ L \mid L \in ((\rationalset_1 \cup \rationalset_2) -
      (\rationalset_1 \cap \rationalset_2)) \}$\\			
      \bottomrule
    \end{tabularx}
   }
  \end{center}
\end{definition}

\paragraph*{Language Restrictions.}
\label{sec:lang-restr}
We analyze three different classes of \RegularlyGeneratedLanguageSetsAbbrev for being closed
under these operators:
\begin{inparaenum}[\bfseries(1)]
\item General \RegularlyGeneratedLanguageSetsAbbrev,
\item finite \RegularlyGeneratedLanguageSetsAbbrev, and
\item finite \RegularlyGeneratedLanguageSetsAbbrev with a fixed
  language substitution $\varphi$.
\end{inparaenum}
For closure properties, we do \emph{not distinguish} between Kleene
star free and finite \RegularlyGeneratedLanguageSetsAbbrev, since
every finite \RegularlyGeneratedLanguageSetAbbrev is expressible as
Kleene star free \RegularlyGeneratedLanguageSetAbbrev (however, given
a \RegularlyGeneratedLanguageSetAbbrev with Kleene star, it is
non-trivial to decide whether the given
\RegularlyGeneratedLanguageSetAbbrev it finite or not~\cite{membership}).
Therefore, all closure properties for finite
\RegularlyGeneratedLanguageSetsAbbrev apply to Kleene star free
\RegularlyGeneratedLanguageSetsAbbrev as well.
Hence, cases {\bfseries(2-3)} correspond to \FQL.
Case {\bfseries(3)} is relevant for usability in practice, allowing to
apply the corresponding operators without constructing a new language
substitution. 
This does not only significantly reduce the search space but also
provides more intuitive results to users.

\begin{theorem}[Closure Properties of \RegularlyGeneratedLanguageSetsTitleAbbrev]
  \label{thm:closure-properties}
  %
  The following Table summarizes the closure properties for
  \RegularlyGeneratedLanguageSetsAbbrev.
  %
  %
  \begin{center}
  	{
	  	\relsize{-1}
	    \begin{tabularx}{\textwidth}{p{9.5em}r@{\hspace{.5em}}CCZ}
	      \toprule
	      Operation & & \multicolumn{3}{c}{Closure Property}\\\cmidrule{3-5}
	               & & General  & \multicolumn{2}{c}{Finite \RegularlyGeneratedLanguageSetsAbbrev}\\\cmidrule{4-5}
	        \multicolumn{3}{l}{\relsize{-1}(+ closed\ \ \ \ - not closed\
              \ \ \ ? unknown)}             & General & Fixed Subst.\\
	      \midrule
	      Product     & Prop.~\ref{prop:closure:product-kleene} & +       & +       & + \\
	      Kleene Star & Prop.~\ref{prop:closure:product-kleene}    & +       & -       & - \\
	      \hspace*{2em}Point-wise & Prop.~\ref{prop:closure:pw-star}       & -       & +       & - \\
	      
	      Complement              & Prop.~\ref{prop:closure:complement}      & -       & -       & - \\
	      \hspace*{2em}Point-wise & Prop.~\ref{prop:closure:pw-complement}    & -       & +       & - \\
	
	      Union         & Prop.~\ref{prop:closure:union}        & +       & +       & + \\
	      \hspace*{2em}Point-wise & Prop.~\ref{prop:closure:pw-union}      & -       & +       & - \\
	      \hspace*{2em}Cartesian  &  Cor.~\ref{cor:catesian-binary}  & -       & +       & - \\
	      
	      Intersection  & Prop.~\ref{prop:closure:intersection}    & ?       & +       & + \\
	      \hspace*{2em}Point-wise & Prop.~\ref{prop:closure:pw-intersection}   & -       & +       & - \\
	      \hspace*{2em}Cartesian  &  Cor.~\ref{cor:catesian-binary} & -       & +       & - \\
	      
	      Difference    & Prop.~\ref{prop:closure:difference}     & ?       & +       & + \\
	      \hspace*{2em}Point-wise & Prop.~\ref{prop:closure:pw-difference}    & -       & +       & - \\
	      \hspace*{2em}Cartesian  & Cor.~\ref{cor:catesian-binary}   & -       & +       & - \\
	      \hspace*{2em}Symmetric  & Prop.~\ref{prop:closure:sym-difference}   & ?       & +       & + \\
	      \bottomrule
	    \end{tabularx}
    }
  \end{center}
\end{theorem}

\paragraph*{Unifying Alphabets.}
\label{sec:unifying-alphabets}
Let $\rationalset_1$ and $\rationalset_2$ be \RegularlyGeneratedLanguageSetsAbbrev over a
common alphabet $\Sigma$ with $\rationalset_i = (K_i, \varphi_i)$,
$K_i \subseteq \Delta_i$, and $\varphi_i : \Delta_i \rightarrow
2^{\Sigma^*}$.
Then we create a unified alphabet $\Delta = \{\left<i, \delta\right> \mid \delta
\in \Delta_i \mbox{ with } i=1,2 \}$ and a unified language
substitution $\varphi : \Delta \rightarrow 2^{\Sigma^*}$ with
$\varphi(\left<i, \delta\right>) = \varphi_i(\delta)$.
We obtain $\rationalset_i=(K'_i, \varphi)$ where $K'_i$ is derived
from $K_i$ by substituting each symbol $\delta\in\Delta_i$ with
$\left<i,\delta\right>\in\Delta$.
Hence without loss of generality, we \textbf{fix} the
\textbf{alphabets} $\Delta$ and $\Sigma$ with \textbf{language
  substitution}~$\varphi$, allowing our \RegularlyGeneratedLanguageSetsAbbrev only to differ
in the generating languages $K_i$.
When we discuss binary operators, we freely refer to \textbf{\RegularlyGeneratedLanguageSetsAbbrev} $\rationalset_i=(K_i,\varphi)$ for $i=1,2$, in case of unary
operators to $\rationalset=(K,\varphi)$, and in case of point-wise
operators to the regular language $R\subseteq\Sigma^\star$.

\paragraph*{General Observations.}
We exploit in our proofs some general observations on the cardinality
of \RegularlyGeneratedLanguageSetsAbbrev.
Moreover, we prove all closure properties of Cartesian binary
operators by reducing the point-wise operators to the Cartesian
one.
For space reasons, we show this generic argument only in
Appendix~\ref{sec:proofs:cart-binary-oper}.

\begin{fact}[Finite Sets of Regular Languages are \ARegularlyGeneratedLanguageSetTitle]
  \label{fact:finite-regular-sets-are-rational}
  Every finite set of regular languages is \ARegularlyGeneratedLanguageSet. 
\end{fact}

\begin{proof}
  For a finite set of regular languages $\rationalset$, we set
  $\varphi(\delta_L)=L$ for all $L\in\rationalset$, taking fresh
  symbols $\delta_L$.
  With $\Delta_\rationalset=\{\delta_L\mid L\in\rationalset\}$ we
  obtain $\rationalset=(\Delta_\rationalset,\varphi)$.\qed
\end{proof}

\begin{fact}[Cardinality of \RegularlyGeneratedLanguageSetsTitleAbbrev]
  \label{fact:rational-sets-cardinality}
  A \RegularlyGeneratedLanguageSetAbbrev contains at most countably many languages. In
  particular, $2^{\Sigma^*}$ is not a \RegularlyGeneratedLanguageSetAbbrev. 
\end{fact}

\begin{proof}
  A \RegularlyGeneratedLanguageSetAbbrev $\rationalset = (K, \varphi)$ is countable, as $K$
  contains countably many words, and $|K|\ge |\rationalset|$ holds.
  Since $2^{\Sigma^*}$ is uncountable, it is not a \RegularlyGeneratedLanguageSetAbbrev. \qed
\end{proof}

\subsection{Product and Kleene Star}
\label{sec:product-kleene-star}

\begin{proposition}[Closure of Product and Kleene Star]
  \label{prop:closure:product-kleene}
  \begin{inparaenum}[\normalfont\bfseries(1)]
  \item $\rationalset_1 \cdot \rationalset_2$ is a \RegularlyGeneratedLanguageSetAbbrev,
    defined over the same substitution $\varphi$. If $\rationalset_i$
    are finite, then $\rationalset_1 \cdot \rationalset_2$ is also
    finite.
  \item $\rationalset^\star$ is a \RegularlyGeneratedLanguageSetAbbrev.
    It is in general infinite even if $\rationalset$ is finite.
  \end{inparaenum}
\end{proposition}

\begin{proof}
  \begin{inparaenum}[\bfseries(1)]
  \item We construct $\rationalset' = (K', \varphi)$ with $K' = K_1
    \cdot K_2$ and obtain $\rationalset_1 \cdot \rationalset_2 =
    \rationalset'$.
  \item We construct $\rationalset' = (K', \varphi')$ with $K' =
    K^\star \setminus \{ \varepsilon \} \cup \{ \delta_\varepsilon \}$
    setting $\varphi'(\delta_\varepsilon) = \{ \varepsilon \}$ and
    $\varphi'(\delta) = \varphi(\delta)$ otherwise, and obtain
    $\rationalset^\star=\rationalset'$.
    Consider the finite \RegularlyGeneratedLanguageSetAbbrev $\rationalset = \{ \{ a \} \}$,
    then, $\rationalset^\star$ is the infinite \RegularlyGeneratedLanguageSetAbbrev $\{ \{ a^i
    \} \mid i \geq 0 \}$.
    \qed
  \end{inparaenum}
\end{proof}

In the following we consider the set $S(L)$ of \emph{shortest words}
of a language $L$, disregarding $\varepsilon$, defined with
$S(L)=\{w\in L\setminus\{\varepsilon\} \mid \not\exists w'\in
L\setminus\{\varepsilon\} \mbox{ with } |w'|<|w|\}$.
We also refer to the shortest words $S(\rationalset)$ of a \RegularlyGeneratedLanguageSetAbbrev $\rationalset$ with
$S(\rationalset)=\bigcup_{L\in\rationalset}S(L)$. 

\begin{lemma}
  \label{lemma:minimal-words}
  Let $\varepsilon \in \varphi(\delta)$ hold for all $\delta \in
  \Delta$.
  Then, for each $w \in \Delta^+$ and shortest word $v\in
  S(\varphi(w))$, there exists a $\delta\in\Delta$ such that $v\in
  S(\varphi(\delta))$.
\end{lemma}

\begin{proof}
  We start with a little claim: 
  Because of $\varepsilon\in\varphi(\delta)$ for all
  $\delta\in\Delta$, we have $\varphi(\delta_i)\subseteq \varphi(w)$
  for $w=\delta_1\dots \delta_k$ and all $1\le i \le k$.

  Assume $v\in S(\varphi(w))$ with $v\not\in\varphi(\delta)$ for all
  $\delta\in\Delta$.
  Then $v=v_1\dots v_k$ with $v_i\in\varphi(\delta_i)$, and since
  $v\neq \varepsilon$, $v_p\neq\varepsilon$ for some $1\le p\le k$. We
  fix such a $p$.
  From the claim above, we get $v_p\in\varphi(\delta_p)\subseteq
  \varphi(w)$, leading to a contradiction:
  If $v\neq v_p$, then $v$ is not a shortest word in
  $\varphi(w)\setminus\{\varepsilon\}$, as $v_p$ is shorter. 
  If $v=v_p$, we contradict our assumption with
  $v=v_p\in\varphi(\delta_p)$.

  Thus, we have shown that there exists a $\delta$ with
  $v\in\varphi(\delta)$.
  It remains to show $v\in S(\varphi(\delta))$.
  Assuming that $v'\in\varphi(\delta)\setminus\{\varepsilon\}$ is
  shorter than $v$, we quickly arrive at a contradiction:
  $v'\in\varphi(\delta)\subseteq \varphi(w)$ from the claim above,
  implies that $v$ would not be a shortest word in
  $\varphi(w)\setminus\{\varepsilon\}$ in the first place, i.e.,
  $v\not\in S(\varphi(w))$. 
  \qed
\end{proof}

\begin{corollary}
  \label{cor:finite-number-of-shortest-words}
  Let $\varepsilon \in \varphi(\delta)$ hold for all $\delta \in
  \Delta$.
  Then the set of shortest words $S(\rationalset)$ is finite.
\end{corollary}

\begin{proof}
  Lemma~\ref{lemma:minimal-words} states for each word $v\in
  S(\rationalset)$, we have $v\in S(\varphi(\delta))$ for some
  $\delta\in\Delta$.
  But there are only finitely many symbols $\delta\in\Delta$, each
  generating only finitely many shortest words in
  $\varphi(\delta)\setminus\{\varepsilon\}$. 
  Hence $S(\rationalset)$ must be finite.\qed
\end{proof}

\begin{proposition}[Closure of Point-wise Kleene Star]
  \label{prop:closure:pw-star}
  \begin{inparaenum}[\normalfont\bfseries(1)]
  \item\label{prop:closure:pw-star:a} In general,
    $\dot{\rationalset}^\star$ is not a \RegularlyGeneratedLanguageSetAbbrev.
  \item\label{prop:closure:pw-star:b} If $\rationalset$ is finite,
    $\dot{\rationalset}^\star$ is a finite \RegularlyGeneratedLanguageSetAbbrev.
  \item\label{prop:closure:pw-star:c} In the latter case, expressing
    $\dot{\rationalset}^\star$ requires a new language
    substitution~$\varphi$.
  \end{inparaenum}
\end{proposition}

\begin{proof}
  \begin{inparaenum}[\bfseries(1)]
  \item Consider the \RegularlyGeneratedLanguageSetAbbrev $\rationalset = \{ \{ a^i \} \mid i
    \geq 1 \}$ with $\dot{\rationalset}^\star = \{ L_i \mid i \geq 1
    \}$ with $L_i=\{ a^{j \cdot i} \mid j \geq 0 \}$.
    Every language $L_i \in \dot{\rationalset}^\star$ contains the
    empty word~$\varepsilon=a^{0 \cdot i}$, and hence,
    $\varepsilon\in\varphi(\delta)$ for all $\delta\in\Delta$
    (disregarding symbols $\delta$ not occurring in $K$).
    Thus, Corollary~\ref{cor:finite-number-of-shortest-words} applies,
    requiring that the set of shortest words
    $S(\dot{\rationalset}^\star)$ is finite. 
    This leads to a contradiction, since
    $S(\dot{\rationalset}^\star)=\{a^i|i\ge 1\}$ is infinite. 
  \item Since $\rationalset$ is finite, also
    $\dot{\rationalset}^\star$ has to be finite and statement follows
    from Fact~\ref{fact:finite-regular-sets-are-rational}.
  \item Consider the \RegularlyGeneratedLanguageSetAbbrev $\rationalset = \{\{a\}\}$, produced
    from $(K,\varphi)$ with $K=\{\delta_a\}$ and
    $\varphi(\delta_a)=a$.
    Then, $\dot{\rationalset}^\star = \{\{ a^i \mid i \geq 0 \}\}$,
    and since $\{ a \} \neq \{ a^i \mid i \geq 0 \}$ we have to
    introduce a new symbol.  \qed
  \end{inparaenum}
\end{proof}


\subsection{Complement}
\label{sec:complement}

\begin{proposition}[Non-closure under Complement]
  \label{prop:closure:complement}
  Let $\rationalset$ be a \RegularlyGeneratedLanguageSet.  Then $\overline{\rationalset}$
  is not a \RegularlyGeneratedLanguageSet.
\end{proposition}

\begin{proof}
  Fact~\ref{fact:rational-sets-cardinality} states that $\rationalset$
  is countable while $2^{\Sigma^*}$ is uncountable.
  Hence, $2^{\Sigma^*}~\setminus~\rationalset$ is uncountable and is
  therefore inexpressible as \RegularlyGeneratedLanguageSetAbbrev.\qed
\end{proof}

\begin{proposition}[Closure of Point-wise Complement]
  \label{prop:closure:pw-complement}
  \begin{inparaenum}[\normalfont\bfseries(1)]
  \item\label{prop:closure:pw-complement:a}
    $\dot{\overline{\rationalset}}$ is in general not a \RegularlyGeneratedLanguageSetAbbrev.
  \item\label{prop:closure:pw-complement:b} If $\rationalset$ is
    finite, $\dot{\overline{\rationalset}}$ is a finite \RegularlyGeneratedLanguageSetAbbrev
    as well,
  \item requiring, in general, a modified language substitution.
  \end{inparaenum}
\end{proposition}

\begin{proof}
  \begin{inparaenum}[\bfseries(1)]
  \item 
    %
    Consider the \RegularlyGeneratedLanguageSetAbbrev $\rationalset = (K, \varphi)$ with $K =
    L(\delta\delta^*)$ and $\varphi(\delta) = \{ a, b \} = \Sigma$.
    Then we have $\rationalset=\{\Sigma^i|i\ge 1\}$. 
    For $i \neq j$, we have $\overline{\Sigma^i} \cap
    \overline{\Sigma^j} = \emptyset$, and consequently,
    $\overline{\Sigma^i} \not\subseteq \overline{\Sigma^j}$ and
    $\overline{\Sigma^i} \not\supseteq \overline{\Sigma^j}$.
    Furthermore, observe $\varepsilon \in \overline{\Sigma^i}$ for
    each $i \geq 1$.
    Assume $\dot{\overline{\rationalset}}$ is a \RegularlyGeneratedLanguageSetAbbrev.  Then,
    there are $K'$ and $\varphi'$ such that
    $\dot{\overline{\rationalset}} = (K', \varphi')$.  
    Since $\rationalset'$ is infinite and $K'$ is regular, there
    exists a word $w \in K'$ with $w = uvz$ and $\varphi(v) \neq \{
    \varepsilon \}$ and $uv^iz \in K'$ for all $i \geq 1$.
    Because of $\varepsilon \in \overline{\Sigma^p}= \varphi(uvz)$ for
    some $p$, we obtain $\varepsilon \in \varphi(v)$ as well.
    But then, for all $i \geq 1$, $\varphi(uvz) \subseteq \varphi(uv^iz)$,
    i.e., $\varphi(uvz)=\overline{\Sigma^p} \subseteq
    \overline{\Sigma^q}=\varphi(uv^iz)$.
    This contradicts the observation that $\overline{\Sigma^p}
    \not\subseteq \overline{\Sigma^q}$.
  \item By Fact~\ref{fact:finite-regular-sets-are-rational}.
  \item Let $\rationalset = (\{ \delta_a \}, \varphi)$ with
    $\varphi(\delta_a) = \{ a \}$.  Then,
    $\dot{\overline{\rationalset}} = \{ \Sigma^* \setminus \{a \} \}$.
    But, $\{ a \} \neq \Sigma^* \setminus \{a \}$.  Therefore, we need
    a new symbol to represent $\Sigma^* \setminus \{ a \}$.  \qed
  \end{inparaenum}
\end{proof}

In contrast to complementation, some \RegularlyGeneratedLanguageSetsAbbrev have a point-wise
complement which is a \RegularlyGeneratedLanguageSetAbbrev as well; first, this is true for
all finite \RegularlyGeneratedLanguageSetsAbbrev, as shown above, but there are also some
infinite \RegularlyGeneratedLanguageSetsAbbrev which have point-wise complement.

\begin{example}
  The \RegularlyGeneratedLanguageSetAbbrev $\rationalset = (L(\delta\delta^*), \varphi)$ with
  $\varphi(\delta) = \{a, b, \varepsilon\}$ has the point-wise
  complement
  $\dot{\overline{\rationalset}}=(L(\delta_1\delta_1\delta_1^*\delta_2),
  \varphi')$ with $\varphi'(\delta_1) = \{a, b\}$ and
  $\varphi'(\delta_2) = L((a + b)^*)$.
\end{example}


\subsection{Union}
\label{sec:union}

\begin{proposition}[Closure of Union]
  \label{prop:closure:union}
  The set $\rationalset_1 \cup \rationalset_2$ is a \RegularlyGeneratedLanguageSet,
  expressible as $(K_1 \cup K_2, \varphi)$ without changing the
  substitution~$\varphi$.
\end{proposition}

\begin{proof}
  Regular languages are closed under union, hence the claim
  follows.\qed
\end{proof}

The following set of regular languages is not \ARegularlyGeneratedLanguageSet.
We will use it in the proof of Proposition~\ref{prop:closure:pw-union} to show that, in general, \RegularlyGeneratedLanguageSetsAbbrev are not closed under point-wise union.
\begin{example}
	\label{expl:pumping-lemma}
	Consider the set $\cM = \{ \{ b \} \cup \{ a^{i} \mid 1 \leq i \leq n + 1 \} \mid n \in \bN \} \subseteq 2^{\{a, b\}^*}$.
	$\cM$ contains infinitely many languages, therefore, any \RegularlyGeneratedLanguageSetAbbrev $\rationalset = (K, \varphi)$, with $\cM = \cR$, requires a regular language $K$ containing infinitely many words.
	By $L_n$ we denote the set $\{b\} \cup \{ a^i \mid 1 \leq i \leq n + 1\}$.
	Then, $L_0 \subsetneq L_1 \subsetneq \ldots L_{i - 1} \subsetneq L_i \subsetneq L_{i + 1} \subsetneq \ldots$.
	There must be a word $w = uvz \in K$ such that $uv^iz\in K$, for all $i \geq
  1$ (cf.~pumping lemma for regular languages~\cite{DBLP:books/aw/HopcroftU79}).
	Furthermore, there must be such a word $w = uvz$ such that $\varphi(u) \neq \emptyset$, $\varphi(v) \neq \emptyset$, $\varphi(v) \neq \{ \varepsilon \}$, and $\varphi(z) \neq \emptyset$.
	This is due to the fact that we have to generate arbitrary long words $a_i$.
	We can assume that $b \notin \varphi(v)$ because otherwise $b^i \in \varphi(v^i)$, for all $i \geq 1$.
	Therefore, $a^k \in \varphi(v)$ for some $k \geq 1$.
	Since $b \in \varphi(uvz)$ has to be true, we can assume w.l.o.g.~that $b \in \varphi(u)$.
	But, then $ba^k\ldots \in \varphi(uvz)$.
	This is a contradiction to the fact that, for all $n \geq 1$, $ba^k\ldots \notin L_n$.
\end{example}

\begin{proposition}[Closure of Point-wise Union]
  \label{prop:closure:pw-union}
  \begin{inparaenum}[\normalfont\bfseries(1)]
  \item\label{prop:closure:pw-union:a} The set $\rationalset_1 \dotcup
    R$ is, in general, not a \RegularlyGeneratedLanguageSetAbbrev.
  \item\label{prop:closure:pw-union:b} The set $\rationalset \dotcup
    R$ is a \RegularlyGeneratedLanguageSetAbbrev for finite $\rationalset$.
  \item\label{prop:closure:pw-union:c} In the latter case, the
    resulting \RegularlyGeneratedLanguageSetAbbrev requires in general a different language
    substitution.
  \end{inparaenum}
\end{proposition}

\begin{proof}
  \begin{inparaenum}[\bfseries(1)]
  \item Let $\rationalset = (L(\delta_1\delta_2^*), \varphi)$ with
    $\varphi(\delta_1) = \{ a \}$ and $\varphi(\delta_2) = L(a +
    \varepsilon)$ and let $R = \{ b \}$.  Then, $\rationalset \dotcup
    R = \{ \{ b \} \cup \{ a^i \mid 1 \leq i \leq n + 1 \} \mid n \in
    \bN \}$ which is not a \RegularlyGeneratedLanguageSetAbbrev, as shown in
    Example~\ref{expl:pumping-lemma}.
  \item By Fact~\ref{fact:finite-regular-sets-are-rational}.
  \item\label{proof:star-free-fixed-alphabet-point-wise-union} Let
    $\rationalset = (\{ \delta \}, \varphi)$ with $\Delta = \{
    \delta \}$, $\Sigma = \{ a, b \}$, $\varphi(\delta) = \{ a \}$
    and let $R = \{ b \}$.  Then, $\rationalset \dotcup R = \{ \{ a,
    b \} \}$, which is inexpressible with $\varphi$. \qed
  \end{inparaenum}
\end{proof}



\subsection{Intersection}
\label{sec:intersection}

\begin{proposition}[Closure of Intersection]
	\label{prop:closure:intersection}
    	Let $\rationalset_1$ and
      $\rationalset_2$ be two finite \RegularlyGeneratedLanguageSetsAbbrev using the same
      language substitution~$\varphi$.  Then, $\rationalset_1 \cap
      \rationalset_2$ is a finite \RegularlyGeneratedLanguageSetAbbrev which can be expressed
      using the language substitution~$\varphi$.
\end{proposition}

\begin{proof}
  We can enumerate each word $w_1 \in K_1$ and check whether there is a word $w_2 \in K_2$ such that $\varphi(w_1) = \varphi(w_2)$.
	If so, we keep $w_1$ in a new set $K_3 = \{ w_1 \in K_1 \mid \exists w_2 \in K_2. \varphi(w_1) = \varphi(w_2)\}$ and $(K_3, \varphi) = \rationalset_1 \cap \rationalset_2$.
	\qed
\end{proof}

In general, \RegularlyGeneratedLanguageSetsAbbrev are not closed under point-wise intersection
but they are closed under point-wise intersection when restricting to
finite \RegularlyGeneratedLanguageSetsAbbrev.

\begin{proposition}[Closure of Point-wise Intersection]
  \label{prop:closure:pw-intersection}
  \begin{inparaenum}[\normalfont\bfseries(1)]
  \item\label{prop:closure:pw-intersection:a} \RegularlyGeneratedLanguageSetsStartAbbrev are not
    closed under point-wise intersection.
  \item\label{prop:closure:pw-intersection:b} For finite
    \rationalset $\rationalset \dotcap R$ is a finite \RegularlyGeneratedLanguageSetAbbrev,
  \item\label{prop:closure:pw-intersection:c} in general requiring a
    different language substitution.
  \end{inparaenum}
\end{proposition}

\begin{proof}
  \begin{inparaenum}[\bfseries(1)]
  \item Let $\rationalset = (K, \varphi)$ with $K = L(\delta\delta^*)$
    and $\varphi(\delta) = L({\tt a + b^\star})$, and set $R = L({\tt
      a^\star + b})$.  
    Then $\rationalset \dotcap R = \{ \{ b \} \cup \{ a^{i} \mid 1
    \leq i \leq n + 1 \} \mid n \in \bN \}$.  
    In Example~\ref{expl:pumping-lemma}, we showed that
    $\rationalset\dotcap R$ is not a \RegularlyGeneratedLanguageSetAbbrev.
  \item By Fact~\ref{fact:finite-regular-sets-are-rational}.
  \item Let $\rationalset = (K, \varphi)$ with $K = \{ \delta \}$ and
    $\varphi(\delta) = L({\tt a + b^\star})$, and set $R =
    L({\tt a^\star + b})$.
    Then, $\rationalset \dotcap R = \{ L({\tt a + b }) \}$ which in
    inexpressible via $\varphi$.  \qed
  \end{inparaenum}
\end{proof}



\subsection{Set Difference}
\label{sec:set-difference}

\begin{proposition}[Closure of Difference]
  \label{prop:closure:difference}
  	For finite $\rationalset_1$
    and $\rationalset_2$, $\rationalset_1 - \rationalset_2$ is a
    finite \RegularlyGeneratedLanguageSetAbbrev, expressible as $(K_3, \varphi)$, for some
    $K_3 \subseteq K_1$.
\end{proposition}

\begin{proof}
  	Set $K_3=\{w\in K_1 \mid \varphi(w)\in \rationalset_2\}$ and
    the claim follows. \qed
\end{proof}

\begin{proposition}[Closure of Point-wise Difference]
  \label{prop:closure:pw-difference}
  \begin{inparaenum}[\normalfont\bfseries(1)]
  \item\label{prop:closure:pw-difference:a} In general, $\rationalset
    \dotdiv R$ is not a \RegularlyGeneratedLanguageSetAbbrev.
  \item\label{prop:closure:pw-difference:b} $\rationalset \dotdiv R$
    is a finite \RegularlyGeneratedLanguageSetAbbrev for finite \rationalset, 
  \item\label{prop:closure:pw-difference:c} requiring in general a
    different language substitution.
  \end{inparaenum}
\end{proposition}

\begin{proof}
  \begin{inparaenum}[\bfseries(1)]
  \item Let $\rationalset = (L(\delta_1\delta_2^*), \varphi)$ with
    $\varphi(\delta_1) = L(a + b)$ and $\varphi(\delta_2) = L(a + b +
    \varepsilon)$.  Let $R = L(bbb^* + (a + b)^*ab(a + b)^* + (a +
    b)^*ba(a + b)^*)$.  Then, $\rationalset \dotdiv R = \{ \{ b \}
    \cup \{ a^i \mid 1 \leq i \leq n + 1\} \mid n \in \bN\}$ which is
    not a \RegularlyGeneratedLanguageSetAbbrev (see Example~\ref{expl:pumping-lemma}).
  \item By Fact~\ref{fact:finite-regular-sets-are-rational}.
  \item Let $\rationalset = (\{\delta_a\}, \varphi)$ with
    $\varphi(\delta_a) = \{ a \}$ and let $R = \{ a \}$.  Then,
    $\rationalset \dotdiv R = \{ \emptyset \}$, requiring a new
    symbol.\qed
	\end{inparaenum}
\end{proof}


\begin{proposition}[Closure of Symmetric Difference]
  \label{prop:closure:sym-difference}
  	Let $\rationalset_1$ and
    $\rationalset_2$ be finite \RegularlyGeneratedLanguageSetsAbbrev using the same language
    substitution~$\varphi$.  Then,
    $\rationalset_1\Delta\rationalset_2$ is a finite \RegularlyGeneratedLanguageSetAbbrev and
    can be expressed using the language substitution~$\varphi$.
\end{proposition}

\begin{proof}
	The proof follows immediately from the closure properties of union, intersection, and difference.
	\qed
\end{proof}

\section{Decision Problems}
\label{sec:decision-problems}

Given a a regular language $R\subseteq\Sigma^\star$ and a
\RegularlyGeneratedLanguageSetAbbrev $\rationalset = (K, \varphi)$
over the alphabets $\Delta$ and $\Sigma$, the \emph{membership
  problem} is to decide whether $R\in\rationalset$ holds.
Given another $\rationalset' = (K', \varphi')$, also over the
alphabets $\Delta'$ and $\Sigma$, the \emph{inclusion problem} asks
whether $\rationalset\subseteq\rationalset'$ holds, and the
\emph{equivalence problem,} whether $\rationalset=\rationalset'$
holds.

\begin{theorem}[Equivalence, Inclusion, and Membership for Kleene star
  free \RegularlyGeneratedLanguageSetsAbbrev]
  \label{thm:kleene-star-free}
  Membership, inclusion, and equivalence are \PSPACE-complete for
  Kleene star free \RegularlyGeneratedLanguageSetsAbbrev. 
\end{theorem}

This holds true, since in case of Kleene star free
\RegularlyGeneratedLanguageSetsAbbrev, we can enumerate the regular
expressions defining all member languages in \PSPACE.
Given the \PSPACE-completeness of regular language equivalence, we
compare a given regular expression with all member languages, solving
the membership problem in \PSPACE.
Doing so for all languages of another
\RegularlyGeneratedLanguageSetAbbrev solves the inclusion problem, and
checking mutual inclusion yields an algorithm for equivalence.
This approach does \emph{not} immediately generalize to finite
\RegularlyGeneratedLanguageSetsAbbrev, since finite
\RegularlyGeneratedLanguageSetsAbbrev $\rationalset=\{\varphi(w)\mid
w\in K\}$ may be generated from an infinite $K$ with Kleene stars.

In the general case, the situation is quite different:
Previous work shows that the membership problem is
decidable~\cite{membership}.
Taking this work as starting point, we give a first \tEXPSPACE upper
bound on the complexity of the problem.
A corresponding lower bound is missing, however we expect the problem
to be at least \EXPSPACE-hard.
Due to space reasons, we only give an overview on the algorithm in the
paper and must defer its details to the appendix.
Finally, the decidability of inclusion and equivalence are still open
problems.
%

\subsection{Membership for general \RegularlyGeneratedLanguageSetsAbbrev}
\label{sec:memb-gener-case}

\begin{algorithm}[tbp]
  \Input{regular languages $R\subseteq\Sigma^\star$, $K\subseteq\Delta^\star$,\\
    \hspace*{4mm} regular language substitution $\varphi$ with
    $\varphi(\delta)\subseteq\Sigma^\star$ for all $\delta\in\Delta$, and\\
    \hspace*{4mm} all as regular expressions}
  \Returns{$\boolTrue$ iff $\exists w\in K: \varphi(w)=R$ (i.e., iff $R\in(K,\varphi)$)}

  \ForEach{$M' \in \setenumeration(R,K,\varphi)$}{ \label{line:mem:decomp-iteration}
    \lIf{$\basiccheck(R,M',\varphi)$}{\label{line:mem:check}
      \Return $\boolTrue$\;}}
  \Return $\boolFalse$\;
  \caption{$\membership(R,K,\varphi)$}
  \label{alg:membership}
\end{algorithm}

By definition, the membership problem is equivalent to asking whether
there exists a $w\in K$ with $\varphi(w)=R$.
For checking the existence of such a $w$, we have to check possibly
infinitely many words in $K$ efficiently.
\begin{inparaenum}[\bfseries (A)]
  To render this search feasible, we
\item rule out irrelevant parts of $K$, and
\item treat subsets of $K$ at once.
\end{inparaenum}
This leads to the procedure $\membership(K,R,\varphi)$ shown in
Algorithm~\ref{alg:membership}, which first enumerates with $M' \in
\setenumeration(K,R,\varphi)$ a sufficient set of sublanguages
(Line~\ref{line:mem:decomp-iteration}), and then checks each of those
sublanguages individually (Line~\ref{line:mem:check}).
More specifically, we employ the following optimizations:
\begin{inparaenum}
  We rule out
\item[\bfseries (A.1)] all words $w$ with $\varphi(w)\not\subseteq R$,
  and
\item[\bfseries (A.2)] all words $w$ whose language $\varphi(w)$
  differs from $R$ in the \emph{length of its shortest
    word.} 
  We subdivide the remaining search space 
\item[\bfseries (B)] into finitely many suitable languages $M'$ and
  check the existence of a $w\in M'$ with $\varphi(w)=R$ in a single
  step.
\end{inparaenum}

We discuss a mutually fitting design of these steps below and
consider the resulting complexity. 
However, due to space limitations, we put the details on
$\setenumeration(K,R,\varphi)$ and $\basiccheck(R,M',\varphi)$ into
Sections~\ref{sec:impl-prop} and~\ref{sec:impl-lemma-as},
respectively, followed by Section~\ref{sec:decision-proofs} with the
corresponding proofs. 

\subsubsection{(A.1) Maximal Rewriting.}
\label{sec:maximal-rewriting}
To rule out all $w$ with $\varphi(w)\not\subseteq R$, we rely on the
notion of a \emph{maximal $\varphi$-rewriting $M_\varphi(R)$ of $R$},
taken from~\cite{rewriting}.
$M_\varphi(R)$ consists of the words $w$ with $\varphi(w)\subseteq R$,
i.e., we set $M_\varphi(R) = \{ w \in \Delta^+ \mid \varphi(w)
\subseteq R \}$.
Furthermore, all subsets $M\subseteq M_\varphi(R)$ are called
\emph{rewritings} of $R$, and if $\varphi(M) = R$ holds, $M$ is called
\emph{exact} rewriting.
\begin{proposition}[Regularity of maximal rewritings~\cite{rewriting}]
	\label{def:rewriting}
	Let $\varphi : \Delta \rightarrow 2^{\Sigma^\star}$ be a regular
    language substitution.
    Then the \emph{maximal $\varphi$-rewriting} of a regular language
    $R \subseteq \Sigma^\star$ is a regular language over $\Delta$.
\end{proposition}
As all words $w$ with $\varphi(w)=R$ must be element of
$M_\varphi(R)$, we restrict our search to $M=M_\varphi(R)\cap K$.

\subsubsection{(A.2) Minimal Word Length.}
We restrict the search space further by checking the \emph{minimal
  word length,} i.e., we compare the length of the respectively
shortest word in $R$ and $\varphi(w)$. If $R$ and $\varphi(w)$ have
different minimal word lengths, $R\neq\varphi(w)$ holds, and hence, we
rule out $w$.
We define the minimal word length $\minlength(L)$ of a language~$L$
with $\minlength(L)=\min\{|w| \mid w\in L\}$, leading to the
definition of language strata.
\begin{definition}[Language Stratum]
  \label{def:stratum}
  Let $L$ be a language over $\Delta$, and $\varphi: \Delta\rightarrow
  2^{\Sigma^\star}$ be a regular language substitution, then the
  \emph{$B$-stratum of $L$,} denoted as $L[B, \varphi]$, is the set of
  words in $L$ which generate via $\varphi$ languages of minimal word
  length $B$, i.e., $L[B, \varphi]=\{w\in L \mid
  \minlength(\varphi(w))=B\}$.
\end{definition}
Starting with $M=M_\varphi(R)\cap K$, we restrict our search further
to $M[\minlength(R), \varphi]$.

\subsubsection{(B) 1-Word Summaries.}
\label{sec:1-word-summaries}
It remains to subdivide $M[\minlength(R), \varphi]$ into finitely many
subsets $M'$, which are then checked efficiently without enumerating
their words $w\in M'$.
Here, we only discuss the property of these subsets $M'$ which enables
such an efficient check, and later we will describe an enumeration of
those subsets $M'$.
When we check a subset $M'$, we do not search for a single word $w\in
M'$ with $\varphi(w)=R$ but for a finite set $F\subseteq M'$ with
$\varphi(F)=R$.
The soundness of this approach will be guaranteed by the existence of
\emph{1-word summaries:}
A language $M'\subseteq \Delta^\star$ has 1-word summaries, if for all
finite subsets $F\subseteq M'$ there exists a summary word $w\in M'$
with $\varphi(F)\subseteq\varphi(w)$.
The property we exploit is given by the following proposition.
\begin{proposition}[Membership Condition for Summarizable Languages]
  \label{prop:membership-in-summarizable}
  Let $M'\subseteq \Delta^\star$ be a regular language with 1-word
  summaries and $\varphi(M')\subseteq R$.
  Then there exists a $w\in M'$ with $\varphi(w)=R$ iff there exists a
  finite subset $F\subseteq M'$ with $\varphi(F)=\varphi(M')=R$.
\end{proposition}

\subsubsection{Putting it together.}
\label{sec:putting-it-together}
\begin{inparaenum}[]
\item First, combining \textbf{A.2} and \textbf{B}, we obtain
  Lemma~\ref{lem:representation}, to subdivide the search space $M[B,
  \varphi]$ into a set $\representation(M,B,\varphi)$ of languages $M'$ with
  1-word summaries.
  %
\item Second, in Theorem~\ref{thm:membership-condition}, building upon
  Lemma~\ref{lem:representation} and \textbf{A.1}, we fix
  $B=\minlength(R)$ and iterate through these languages $M'$.
  We check each of them at once with our membership condition from
  Proposition~\ref{prop:membership-in-summarizable}.
  In terms of Algorithm~\ref{alg:membership},
  Lemma~\ref{lem:representation} provides the foundation for
  $\setenumeration(K,R,\varphi)$ and
  Proposition~\ref{prop:membership-in-summarizable} underlies
  $\basiccheck(R,M',\varphi)$.
\end{inparaenum}

\begin{lemma}[Summarizable Language Representation,
  adapting~\cite{membership}]
  \label{lem:representation}
  Let $M\subseteq\Delta^\star$ be a regular language.
  Then, for each bound $B\ge 0$, there exists a family
  $\representation(M,B,\varphi)$ of union-free regular languages
  $M'\in\representation(M,B,\varphi)$ with 1-word summaries, such that
  $M[B, \varphi]\subseteq \bigcup_{M'\in\representation(M,B,\varphi)}
  M'\subseteq M$ holds.
\end{lemma}

\begin{theorem}[Membership Condition, following~\cite{membership}]
  \label{thm:membership-condition}
  Let $\rationalset = (K, \varphi)$ be a
  \RegularlyGeneratedLanguageSetAbbrev and $\varphi : \Delta
  \rightarrow 2^{\Sigma^\star}$ be a regular language substitution.
  Then, for a regular language $R\subseteq\Sigma^\star$, we have
  $R\in\rationalset$,
  iff there exists an $M'\in \representation(M_\varphi(R)\cap
  K,\minlength(R),\varphi)$ with a finite subset $F\subseteq M'$ with
  $\varphi(F)=\varphi(M')=R$.
\end{theorem}

We obtain the space complexity of \membership, depending on the
\emph{size of the expressions,} representing the involved languages.
More specifically, we use the expression sizes $||R||$ and $||K||$ and
the summed size
$||\varphi||=\Sigma_{\delta\in\Delta}||\varphi(\delta)||$ of the
expressions in the co-domain of $\varphi$.

\begin{theorem}[$\membership(R,K,\varphi)$ runs in \tEXPSPACE]
  \label{thm:membership-complexity}
  More precisely, it runs in $\DSPACE\left(||K||^r
    2^{2^{(||R||+||\varphi||)^s}}\right)$ for some constants $r$ and
  $s$.
\end{theorem}

\section{Conclusion}
\label{sec:conclusion}

Motivated by applications in testcase specifications with FQL, we have
studied general and finite
\RegularlyGeneratedLanguageSetsAbbrev. While we showed that general
\RegularlyGeneratedLanguageSetsAbbrev are not closed under most common
operators, \emph{finite} \RegularlyGeneratedLanguageSetsAbbrev are
closed under all operators except Kleene stars and complementation
(Theorem~\ref{thm:closure-properties}).
This shows that our restriction to Kleene star free and hence finite
\RegularlyGeneratedLanguageSetsAbbrev in \FQL results in a natural
framework with good closure properties.
Likewise, the proven \PSPACE-completeness results for Kleene star free
\RegularlyGeneratedLanguageSetsAbbrev provide a starting point to
develop practical reasoning procedures for Kleene star free
\RegularlyGeneratedLanguageSetsAbbrev and \FQL.
Experience with LTL model checking shows that \PSPACE-completeness
often leads to algorithms which are feasible in practice.
In contrast, for general and possibly infinite
\RegularlyGeneratedLanguageSetsAbbrev, we have described a \tEXPSPACE
membership checking algorithm -- leaving the question for matching
lower bounds open.
Nevertheless, reasoning on general
\RegularlyGeneratedLanguageSetsAbbrev seems to be rather infeasible.

Last but not least, \RegularlyGeneratedLanguageSetsAbbrev give rise to new
and interesting research questions, for instance 
the decidability of inclusion and equivalence for general
\RegularlyGeneratedLanguageSetsAbbrev, and the closure properties left open in this paper.
%
%
In our future work, we want to generalize
\RegularlyGeneratedLanguageSetsAbbrev to other base formalisms.
For example, we want $\varphi$ to substitute symbols by context-free
expressions, thus enabling \FQL test patterns to recognize
e.g.~matching of parentheses or emptiness of a stack.  
%

\section*{Acknowledgments}
This work received funding in part by the Austrian National Research
Network S11403-N23 (RiSE) of the Austrian Science~Fund~(FWF), by the
Vienna Science and Technology Fund (WWTF) grant PROSEED, and by the
European Research Council under the European Community's Seventh
Framework Programme (FP7/2007--2013) / ERC grant agreement DIADEM
no.~246858.

\bibliographystyle{splncs}

\appendix
\section{Closure Properties for Cartesian Binary Operators}
\label{sec:proofs:cart-binary-oper}

We deal with Cartesian binary operators generically, by reducing the
point-wise operators to the Cartesian one. 

\begin{lemma}[Reducing Point-Wise to Cartesian Operators]
  \label{lem:reducing-point-wise-to-cartesian}
  Let $\circ$ be an arbitrary binary operator over sets, let $\odot \in \{ \dotcup, \dotcap, \dotdiv \}$, and let $\otimes \in \{ \timescup, \timescap, {\timesminus} \}$.
  \begin{inparaenum}[\normalfont\bfseries(1)]
  \item If $\rationalset_1\odot R$ is not closed under \RegularlyGeneratedLanguageSets,
    then the corresponding $\rationalset_1\otimes \rationalset_2$ is not closed.
  \item If $\rationalset_1\odot R$ is not closed under finite 
  	\RegularlyGeneratedLanguageSets with constant language 
  	substitution, even in presence of a
    symbol $\delta_R$ with $\varphi(\delta_R)=R$, then the corresponding
    $\rationalset_1\otimes\rationalset_2$ is also not closed.
  \end{inparaenum}
\end{lemma}

\begin{proof}
  \begin{inparaenum}[\bfseries(1)]
  \item If $\rationalset_1\odot R$ is not closed, we fix a violating
    pair $\rationalset_1$ and $R$.
    Then we obtain
    $\rationalset_1\otimes\rationalset_2=\rationalset_1\odot R$  for
    $\rationalset_2=(\{\delta_R\},\varphi)$ and $\varphi(\delta_R)=R$.
    Since $\rationalset_1\odot R$ is not a \RegularlyGeneratedLanguageSetAbbrev,
    $\rationalset_1\otimes\rationalset_2$ is not as well, and the
    claim follows.
  \item If $\rationalset_1\odot R$ is inexpressible as a \RegularlyGeneratedLanguageSetAbbrev
    without introducing new symbols in $\varphi$, even in presence of
    $\delta_R$, then $\rationalset_1\otimes\rationalset_2$ is also
    inexpressible without changing $\varphi$.  
    \qed
  \end{inparaenum}
\end{proof}

Given Lemma~\ref{lem:reducing-point-wise-to-cartesian}, it is not
surprising that point-wise and Cartesian operators behave for all
discussed underlying binary operators identically, as shown in Theorem~\ref{thm:closure-properties}.

\begin{corollary}[Closure of Cartesian Binary Operators]
  \label{cor:catesian-binary}
  Let $\otimes \in \{ \timescup, \timescap, {\timesminus} \}$.
  \begin{inparaenum}[\normalfont\bfseries(1)]
  \item The set $\rationalset_1 \otimes \rationalset_2$ is, in
    general, not a \RegularlyGeneratedLanguageSet.
  \item The set $\rationalset_1 \otimes \rationalset_2$ is a
    \RegularlyGeneratedLanguageSet if $\rationalset_1$ and $\rationalset_2$ are finite,
  \item requiring in general a new language substitution.
  \end{inparaenum}
\end{corollary}

\begin{proof}
  \begin{inparaenum}[\bfseries(1)]
  \item By Lemma~\ref{lem:reducing-point-wise-to-cartesian} we reduce
    the point-wise case to the Cartesian case, covered by Propositions
    \ref{prop:closure:pw-union}, \ref{prop:closure:pw-intersection},
    and \ref{prop:closure:pw-difference} for union, intersection, and
    set difference, respectively.
    The claim follows.
  \item Since all considered operators are closed for regular
    languages, the claim follows from
    Fact~\ref{fact:finite-regular-sets-are-rational}.
  \item Again, with Lemma~\ref{lem:reducing-point-wise-to-cartesian}
    we reduce the point-wise case to the Cartesian case.
    The lemma is applicable, as the examples in the proofs of
    Propositions \ref{prop:closure:pw-union},
    \ref{prop:closure:pw-intersection}, and
    \ref{prop:closure:pw-difference} are not jeopardized by a symbol
    $\delta_R$ with $\varphi(\delta_R)=R$.
    Hence the claim follows. \qed
  \end{inparaenum}
\end{proof}

\section{Implementation and Proofs for Section~\ref{sec:decision-problems}}
\label{sec:decision-problems-detail}

\subsection{Implementing $\basiccheck(R,M',\varphi)$}
\label{sec:impl-prop}
Since Lemma~\ref{lem:representation} produces only languages
$M'=N_1S_1^\star N_2\dots N_mS_m^\star N_{m+1}$ with 1-word summaries,
we restrict our implementation to such languages and exploit these
restrictions subsequently.
So, given such a language $M'$ over $\Delta$, and a regular language
substitution $\varphi : \Delta \rightarrow 2^{\Sigma^\star}$, we need to
check whether there exists a finite $F\subseteq M'$ with
$\varphi(F)=\varphi(M')=R$.
We implement this check with the procedure
$\basiccheck(R,M',\varphi)$, splitting the condition of
Proposition~\ref{prop:membership-in-summarizable} into two parts,
namely
\begin{inparaenum}[\bfseries(1)]
\item whether there exists a finite $F\subseteq M'$ with
  $\varphi(F)=\varphi(M')$, and
\item whether $\varphi(M')=R$ holds.
\end{inparaenum}
While the latter condition amounts to regular language equivalence,
the former requires distance automata as additional machinery.

\begin{definition}[Distance Automaton~\cite{membership}]
  \label{def:distance-automaton}
  A \emph{distance automaton} over an alphabet $\Delta$ is a tuple
  $\cA = \langle \Delta, Q, \rho, q_0, F, d\rangle$ where $\langle
  \Delta, Q, \rho, q_0, F\rangle$ is an NFA and $d : \rho
  \rightarrow \{0 , 1\}$ is a distance function, which can be extended
  to a function on words as follows.
  The distance function $d(\pi)$ of a path $\pi$ is the sum of the
  distances of all edges in $\pi$.
  The distance $\mu(w)$ of a word $w \in L(\cA)$ is the minimum of
  $d(\pi)$ for all paths $\pi$ accepting $w$.

  A distance automaton $\cA$ is called \emph{limited} if there exists
  a constant $U$ such that $\mu(w) < U$ for all words $w \in L(\cA)$.
\end{definition}

In our check for {\bfseries(1)}, we build a distance automaton which
is limited iff a finite $F$ with $\varphi(F)=\varphi(M')$ exists.
Then, we rely on the
\PSPACE-decidability~\cite{leung04:_limit_probl_distan_autom} of the
limitedness of distance automata to check whether $F$ exists or not.

\paragraph{Distance-automaton Construction.}
Here, we exploit the assumption that $M'$ is a union-free language
over $\Delta$:
Given the regular expression defining $M'$, we construct the distance
automaton $A_{M'}$ following the form of this regular expression:
\begin{itemize}
\item $\delta\in\Delta$: We construct the finite automaton $A_{\delta}$
  with $L(A_{\delta}) = \varphi(\delta)$.  We extend $A_{\delta}$ to a
  distance automaton by labeling each transition in $A_{\delta_i}$
  with $0$.
\item $e \cdot f$: Given the distance automata $A_{e} = (Q_{e},
  \Sigma, \rho_{e}, q_{0,e}, F_{e}, d_{e})$ and $A_{f} = (Q_{f},
  \Sigma, \rho_{f}, q_{0,f}, F_{f}, d_{f})$, we set $A_{e \cdot f} =
  (Q_{e} \uplus Q_{f}, \Sigma, \rho_{e} \cup \rho_{f} \cup \rho, q_{0,
    e}, F_{f}, d_{e \cdot f})$ where $\rho = \{ (q, \varepsilon, q_{0,
    f}) \mid q \in F_{e} \}$ and $d_{e \cdot f} = d_{e} \cup d_{f}
  \cup \{ (t, 0) \mid t \in \rho \}$, i.e., we connect each final
  state of $A_{e}$ to the initial state of $A_{f}$ and assign the
  distance $0$ to these connecting transitions.
\item $e^\star$: We construct the distance automaton $A_{e} =
  (Q_{e}, \Sigma, \rho_{e}, q_{0, e}, F_{e}, d_{e})$.  Then,
  $A_{e^\star} = (Q_{e}, \Sigma, \rho_{e} \cup \rho, q_{0, e},
  F_{e} \cup \{ q_{0, e} \}, d_{e^\star})$, where $\rho = \{ (q,
  \varepsilon, q_{0, e}) \mid q \in F_{e} \}$ and $d_{e^\star} =
  d_{e} \cup \{ ((q, \varepsilon, p), 1) \mid (q, \varepsilon, p) \in
  \rho \}$, i.e., we connect each final state of $A_{e}$ to the
  initial states of $A_{e}$ and assign the corresponding transitions
  the distance $1$.
\end{itemize}
If the resulting distance automaton $A_{M'}$ is limited, then there
exists a finite subset $F \subseteq M'$ such that $\varphi(F) =
\varphi(M')$.
This implies that {\bfseries(1)} holds.

\begin{algorithm}[tbp]
  \SetKw{Build}{build}
  \Input{regular languages $R\subseteq\Sigma^\star$, $M'\subseteq\Delta^\star$, and\\
    \hspace*{4mm} regular language substitution $\varphi$ with $\varphi(\delta)\subseteq\Sigma^\star$ for all $\delta\in\Delta$,\\
    \hspace*{4mm} all given as regular expressions}
   \Requires{$M'$ is of form $N_1S_1^\star N_2\dots N_mS_m^\star
     N_{m+1}$ with $N_h,S_h\in\Delta^\star$}
  \Requires{$L(M')\subseteq L(R)$}
  \Returns{$\boolTrue$ iff $\exists \mbox{ finite } F\subseteq M':
    \varphi(F)=\varphi(M')=R$}

  \Build $A_{M'}$\;\label{line:bm:limitedautomaton}
  \If{$A_{M'}$ limited}{ \label{line:bm:limitedness}
    \lIf{$\varphi(M') = R$}{ \label{line:bm:equivalence}
      \Return $\boolTrue$;
    }
  }
  \Return $\boolFalse$\;
  \caption{$\basiccheck(R,M',\varphi)$}
  \label{alg:basicmembership}
\end{algorithm}

So, given $M'$, $R$, and all languages in the domain of $\varphi$ as
regular expressions, $\basiccheck(R,M',\varphi)$ in
Algorithm~\ref{alg:basicmembership} first builds $A_{M'}$
(Line~\ref{line:bm:limitedautomaton}) and checks its limitedness
(Line~\ref{line:bm:limitedness}), amounting to condition
{\bfseries(1)}.
For condition {\bfseries(2)}, \basiccheck verifies that $\varphi(M')$
and $R$ are equivalent (Line~\ref{line:bm:equivalence}) and returns
$\boolTrue$ if both checks succeed.

\begin{lemma}[$\basiccheck(R,M',\varphi)$ runs in \PSPACE]
  \label{lem:pspace-completiy-of-prop1-check}
  $\basiccheck(R,M',\varphi)$ runs in \PSPACE, which is optimal, as it
  solves a \PSPACE-complete problem.
\end{lemma}

\subsection{Implementing $\setenumeration(K,R,\varphi)$}
\label{sec:impl-lemma-as}

Our enumeration algorithm must produce the languages
$\representation(M,B,\varphi)$, guaranteeing that all
$M'\in\representation(M,B,\varphi)$ have 1-word summaries, and that
$M[B, \varphi]\subseteq \bigcup_{M'\in\representation(M,B,\varphi)}
M'\subseteq M$ holds (as specified by Lemma~\ref{lem:representation}).
To this end, we rely on a sufficient condition for the existence of
1-word summaries.
First we show this condition with Proposition~\ref{prop:summarizable},
before turning to the enumeration algorithm itself.

\begin{proposition}[Sufficient Condition for 1-Word Summaries]
  \label{prop:summarizable}
  Let $L$ be a union-free language over $\Delta$, given as 
  $L=N_1S_1^\star N_2\dots N_mS_m^\star N_{m+1}$, with words
  $N_h\in\Delta^\star$ 
  and union-free languages $S_h\subseteq\Delta^\star$.
  If $\varepsilon\in\varphi(w)$ for all $w\in S_h$ and all $S_h$, then
  $L$ has 1-word summaries.
\end{proposition}

\begin{algorithm}[tbp]
  \Input{regular languages $R\subseteq\Sigma^\star$,
    $K\subseteq\Delta^\star$,\\
    \hspace*{4mm} regular language substitution $\varphi$ with
    $\varphi(\delta)\subseteq\Sigma^\star$ for all $\delta\in\Delta$, and\\
    \hspace*{4mm} all given as regular expressions}
    \Yields{$L\in\representation(M,\minlength(R),\varphi)$ for $M=M_\varphi(R) \cap K$}
  $M := M_\varphi(R) \cap K$\; \label{line:intersection}
  \lFor{$L\in  \unionfreerep(M)$}{$\unfold(L,\varphi,\minlength(R))$\;}\label{line:iterateunfold}
  \caption{$\setenumeration(R,K,\varphi)$}
  \label{alg:enumeration}
\end{algorithm}

\begin{algorithm}[tbp]
  \Input{regular language $L=N_1S_1^\star N_2\dots N_mS_m^\star N_{m+1}\subseteq\Delta^\star$,\\
    \hspace*{4mm} regular language substitution $\varphi$ with
    $\varphi(\delta)\subseteq\Sigma^\star$ for all $\delta\in\Delta$, and\\
    \hspace*{4mm} a bound $B$} 
  \Yields{$L'\in\representation(L,B,\varphi)$}

  \lIf{$\forall S_h\forall w\in S_h\; : \;
    \varepsilon\in\varphi(w)$}{\Yield{$L$}\;}\label{line:yield}\Else{
    fix $S_h$ arbitrarily with $\exists w\in S_h\; : \; \varepsilon\not\in\varphi(w)$\;\label{line:chooseSH}
    $E:=S_h\cap\Delta^\star_\varepsilon$; \hspace*{4em}\tcp*[h]{$\Delta_\varepsilon=\{\delta\in\Delta\mid\varepsilon\in\varphi(\delta)\}$}\label{line:epsilonlang}\\
    $L_0:=N_1S_1^\star N_2\dots \ N_hE^\star N_{h+1}\dots N_mS_m^\star
    N_{m+1}$\;
    $\unfold(L_0,\varphi,B)$\;\label{line:unfold-rec-0}
    \tcp*[h]{$L_p:=N_1S_1^\star N_2\dots \ N_hE^\star\bar
    E_pS_h^\star N_{h+1}\dots N_mS_m^\star N_{m+1}$ (see text)}\\
    \lFor{$p\in \criticalpos(S_h)$ with $\minlength(\varphi(L_p))\le B$}{$\unfold(L_p,\varphi,B)$\;}\label{line:unfold-rec-iterate}
  }
  \caption{$\unfold(L,\varphi,B)$}
  \label{alg:unfold}
\end{algorithm}

We are ready to design our enumeration algorithm, shown in
Algorithm~\ref{alg:enumeration}, and its recursive subprocedure in
Algorithm~\ref{alg:unfold}.
Both algorithms do not return a result but yield their result as an
enumeration: Upon invocation, both algorithms run through a sequence
of \Yield statements, each time appending the argument of \Yield to
the enumerated sequence.
Thus, the algorithm never stores the entire sequence but only the
stack of the invoked procedures.

Initializing the recursive enumeration,
Algorithm~\ref{alg:enumeration} obtains the maximum rewriting $M :=
M_\varphi(R) \cap K$ of $R$ (Line~\ref{line:intersection}) and
iterates over the languages $L$ in the union-free decomposition of $M$
(Line~\ref{line:iterateunfold}) to call for each $L$ the recursive
procedure $\unfold$, shown in Algorithm~\ref{alg:unfold}.
In turn, Algorithm~\ref{alg:unfold} takes a union free language
$L=N_1S_1^\star N_2\dots N_mS_m^\star N_{m+1}$ and a bound $B$ to
unfold the Kleene-star expressions of $L$ until the precondition of
Proposition~\ref{prop:summarizable} is satisfied or
$\minlength(\varphi(L))>B$.

More specifically, \unfold exploits a rewriting, based on the
following terms:
Given a union free language $S_h$, let
$E=S_h\cap\Delta_\varepsilon^\star$ with $\Delta_\varepsilon =
\{\delta\in\Delta \mid \varepsilon\in\varphi(\delta)\}$ denote all
words $w$ in $S_h$ with $\varepsilon\in\varphi(w)$ and let $\bar
E=S_h\setminus E$.
Since $\bar E$ is in general not union free, we need to split $\bar E$
further. 
To this end, we define $\unfoldstep(S_h,p)$ recursively for an integer
sequence $p=\left<p_H \mid p_T\right>$ with head element $p_H$ and
tail sequence $p_T$.
Intuitively, a sequence $p$ identifies a subexpression in $S_h$ by
recursively selecting a nested Kleene star expression;
$\unfoldstep(S_h,p)$ unfolds $S_h$ such that this selected expression
is instantiated at least once. 
Formally, for $S_h=\alpha_1\beta_1^\star\alpha_2 \dots
\alpha_n\beta_n^\star\alpha_{n+1}$ we set
$\unfoldstep(S_h,\varepsilon)=S_h$ and $\unfoldstep(S_h, p) =
\alpha_1\dots \alpha_{p_H} \beta_{p_H}^\star
\unfoldstep(\beta_{p_H},p_T) \beta_{p_H}^\star
\alpha_{p_H+1}\dots\alpha_{n+1}$.
Consider $S_h=A^\star(B^\star C^\star)^\star D^\star$ (with
all $\alpha_i=\varepsilon$ for brevity), then we obtain 
\vspace*{-.5em}\begin{center}\begin{tabular}{cccccccccccc}
    $\unfoldstep(S_h,\left<2,1\right>)$ & $=$ & $A^\star$ & $(B^\star
    C^\star)^\star$ & \multicolumn{4}{c}{$\unfoldstep(B^\star
      C^\star,\left<1\right>)$} &
    $(B^\star C^\star)^\star$ & $D^\star$ \\
    & $=$ & $A^\star$ & $(B^\star C^\star)^\star$ & $(B^\star$ &
    $\unfoldstep(B,\varepsilon)$ & $B^\star$ & $C^\star)$ & $(B^\star
    C^\star)^\star$ & $D^\star$ \\
    & $=$ & $A^\star$ & $(B^\star C^\star)^\star$ & $(B^\star$ & $(B)$
    & $B^\star$ & $C^\star)$ & $(B^\star C^\star)^\star$ & $D^\star$
\end{tabular}\end{center}\vspace*{-.5em}
instantiating $B$ at position $\left<2,1\right>$ at least once. 
Let $\criticalpos(S_h)$ be integer sequences which identify a
subexpression of $S_h$ which directly contain a symbol $\delta$ with
$\varepsilon\not\in\varphi(\delta)$ (and not only via another
Kleene-star expression).
Then, we write $\bar E= \bigcup_{p\in \criticalpos(S_h)} \bar E_p$,
with $\bar E_p=\unfoldstep(S_h,p)$.
This discussion leads to the following rewriting:

\begin{proposition}[Rewriting for 1-Word Summaries]
  \label{prop:unfold-rewriting}
  For every union free language $S_h^\star$, we have $S_h^\star=
  E^\star\; \cup \; \bigcup_{p\in \criticalpos(S_h)} E^\star\bar E_p
  S_h^\star$. All languages in the rewriting, i.e., $E^\star$ and
  $E^\star\bar E_p S_h^\star$, are union free, $E^\star$ has 1-word
  summaries, and $\minlength(S_h^\star)<\minlength(E^\star\bar E_p
  S_h^\star)$ holds for all $p\in\criticalpos(S_h)$.
\end{proposition}

If $L$ already satisfies the precondition imposed by
Proposition~\ref{prop:summarizable}, Algorithm~\ref{alg:unfold}
\Yield-s $L$ and terminates (Line~\ref{line:yield}).
Otherwise, it fixes an arbitrary $S_h$ violating this precondition and
rewrites $L$ recursively with Proposition~\ref{prop:unfold-rewriting}
(Lines~\ref{line:chooseSH}-\ref{line:unfold-rec-iterate}).
\begin{inparaenum}[\bfseries(1)]
\item \emph{Termination:} In each recursive call, \unfold either
  eliminates in $L_0$ an occurrence of a subexpression $S_h$ violating
  the precondition of Proposition~\ref{prop:summarizable}
  (Line~\ref{line:unfold-rec-0}), or increases the minimum length in
  $L_p$, eventually running into the upper bound $B$
  (Line~\ref{line:unfold-rec-iterate}).
\item \emph{Correctness:} Setting $B=\infty$, \unfold \Yield-s a
  possibly infinite sequence of union free languages which have 1-word
  summaries such that their union equals the original language $L$:
  As the generation of these languages is based on the equality of
  Proposition~\ref{prop:unfold-rewriting} each rewriting step is sound
  and complete, leading to an infinite recursion tree whose leaves
  \Yield the languages in the sequence.
  The upper bound on minimum length only cuts off languages $L_p$
  producing words of minimum length beyond $B$, i.e., $L_p\cap
  L[B,\varphi]=\emptyset$, and in consequence, it is safe to drop
  $L_p$, since we only need to construct
  $\representation(L,B,\varphi)$ with
  $\representation(L,B,\varphi)\supseteq L[B,\varphi]$.
\end{inparaenum}

\subsection{Proofs}
\label{sec:decision-proofs}

\begin{proof}[\textbf{of Theorem~\ref{thm:kleene-star-free}}]
  \emph{\PSPACE-Membership.} We exploit for the \PSPACE-membership of
  all three considered problems the same observations:
  \begin{inparaenum}[\bfseries(1)]
  \item Given Kleene star free languages $K$, we can enumerate in
    \PSPACE all words $w\in K$,
    and
  \item we can check whether $L(R)=L(\varphi(w))$ holds, in
    \PSPACE~\cite{DBLP:conf/focs/MeyerS72}.
  \end{inparaenum}
  
  Thus, to check \emph{membership} of $R$ in $(K,\varphi)$, we
  enumerate all $w\in K$ and check whether $L(R)=L(\varphi(w))$
  holds for some $w$~-- if so, $R\in\rationalset$ is true. 
  For checking the \emph{inclusion}
  $\rationalset'\subseteq\rationalset$, we enumerate all $w'\in K'$
  and search in a nested loop for a $w\in K$ with
  $L(\varphi(w))=L(\varphi(w'))$.
  If such a $w$ exists for all $w'$, we have established
  $(K',\varphi')\subseteq (K,\varphi)$.
  We obtain \PSPACE-membership for \emph{equivalence} $(K',\varphi')=
  (K,\varphi)$ by checking both, $(K',\varphi')\subseteq (K,\varphi)$
  and $(K,\varphi)\subseteq (K',\varphi')$.

  \emph{Hardness.} For hardness we reduce the \PSPACE-complete problem
  whether a given regular expression $X\subseteq\Sigma^\star$ is
  equivalent to $\Sigma^\star$~\cite{DBLP:conf/focs/MeyerS72} to all three
  considered problems:
  Given an arbitrary regular expressions $X$, we set $K=\{a\}$,
  $\varphi(a)=X$, $K'=\{b\}$, $\varphi'(b)=\Sigma^\star$, and
  $R=\Sigma^\star$.
  This gives us $X=\Sigma^\star$ iff $(K,\varphi)=(K',\varphi')$
  (equivalence) iff $(K,\varphi)\subseteq(K',\varphi')$ (inclusion)
  iff $R\in(K,\varphi)$ (membership).
  \qed
\end{proof}

\begin{proof}[\textbf{of Proposition~\ref{prop:membership-in-summarizable}}]
  $(\Rightarrow)$ With $w\in M'$ and $\varphi(w)=R$, taking
  $F=\{w\}\subseteq M'$, we obtain
  $R=\varphi(w)=\varphi(F)\subseteq\varphi(M')\subseteq
  \varphi(M)\subseteq R$, as required.

  $(\Leftarrow)$ $M'$ has 1-word summaries, hence there exists a $w\in
  M'$ with $\varphi(F)\subseteq\varphi(w)$, leading to
  $R=\varphi(F)\subseteq\varphi(w)\subseteq \varphi(M')\subseteq
  \varphi(M)\subseteq R$, as required.  \qed
\end{proof}

\begin{proof}[\textbf{of Lemma~\ref{lem:representation}}]
  We prove the Lemma with Algorithm~\ref{alg:unfold}.
  $\unfold(L,\varphi,B)$ yields $\representation(L,B,\varphi)$ for
  union free languages $L$, hence we obtain $\representation
  (M,B,\varphi)=\bigcup_{L\in\unionfreerep(M)}{\unfold(L,\varphi,B)}$.\qed
\end{proof}

\begin{proof}[\textbf{of Theorem~\ref{thm:membership-condition}}]
  Most of the work for the proof of
  Theorem~\ref{thm:membership-condition} is already achieved by the
  representation $\representation(M,\minlength(R),\varphi)$ of
  Lemma~\ref{lem:representation}:
  The languages $M'\in \representation(M,\minlength(R),\varphi)$ are
  constructed to have 1-word summaries, which make the check whether
  there exists $w \in M'$ with $\varphi(w)=R$ relatively easy~-- this
  is the case iff there exists a finite subset $F\subseteq M'$ with
  $\varphi(F)=\varphi(M')=R$.
  We show both directions of the theorem statement individually.

  $(\Rightarrow)$ Assume $R\in\rationalset$: By Definition~\ref{def:rational-set},
  there exists $w\in K$ with $R=\varphi(w)$, by
  Definition~\ref{def:rewriting}, we get $w\in M_\varphi(R)$, and
  hence $w\in M_\varphi(R)\cap K=M$.
  From $R=\varphi(w)$ and $\minlength(R)=\minlength(\varphi(w))$, we
  get $w\in M[\minlength(R), \varphi]$.
  Since the maximal rewriting $M_\varphi(R)$ of a regular language $R$
  is regular as well~\cite{rewriting}, and since regular languages are
  closed under intersection, we obtain the regularity of $M$, and
  hence, Lemma~\ref{lem:representation} applies.
  Thus, there exists an $M'\in\representation(M,\minlength(R),\varphi)$ with
  $w\in M'$, and via
  Proposition~\ref{prop:membership-in-summarizable}, we obtain 
  for $F=\{w\}\subseteq M'$, $R=\varphi(F)=\varphi(M')$, as required.

  $(\Leftarrow)$ Assume that there exists an $M'\in
  \representation(M,\minlength(R),\varphi)$ with a finite subset $F\subseteq
  M'$ with $\varphi(F)=\varphi(M')=R$.
  Then, via Proposition~\ref{prop:membership-in-summarizable}, we take
  the summary word $w\in M'$ for $F$, yielding $R\in\rationalset$, as
  required. \qed
\end{proof}

\begin{proof}[\textbf{of Lemma~\ref{lem:pspace-completiy-of-prop1-check}}]
  \emph{Membership.} 
  The construction of the automaton $A_{M'}$
  (Line~\ref{line:bm:limitedautomaton}) runs in polynomial time and
  hence produces a polynomially sized distance automaton. 
  Thus, the check for limitedness of $A_{M'}$
  (Line~\ref{line:bm:limitedness}) retains its \PSPACE
  complexity~\cite{DBLP:conf/focs/MeyerS72}.
  Given $M'$, $R$, and all $\varphi(\delta)$ for $\delta\in\Delta$ as
  regular expressions, we can build a polynomially sized regular
  expression for $\varphi(M')$ by substituting $\varphi(\delta)$ for
  each occurrence of $\delta$ in $M'$.
  Then we check the equivalence of the regular expressions for
  $\varphi(M')$ and $R$ (Line~\ref{line:bm:equivalence}), again
  keeping the original \PSPACE complexity of regular expression
  equivalence~\cite{DBLP:conf/focs/MeyerS72}.
  This yields an overall \PSPACE procedure.

  \emph{Hardness.} 
  We reduce the \PSPACE-complete problem of deciding whether a regular
  expression $X$ over $\Sigma$ is equivalent to
  $\Sigma^\star$~\cite{DBLP:conf/focs/MeyerS72} to a single
  \basiccheck invocation~-- proving that \basiccheck solves a \PSPACE
  complete problem.
  Given an arbitrary regular expressions $X$, we set $M'=\{a\}$,
  $\varphi(a)=X$ and $R=\Sigma^\star$. Then
  $\basiccheck(R,M',\varphi)$ returns \boolTrue iff $X$ is equivalent
  to $\Sigma^\star$.
  \qed
\end{proof}

\begin{proof}[\textbf{of Proposition~\ref{prop:summarizable}}]
  We construct  the desired word:
  Choose an arbitrary finite subset $F=\{f_1,\dots,f_p\}\subseteq L$. 
  Then each word $f_i\in F$ is of the form
  $$f_i=N_1s_{1, i}N_2\dots N_ms_{m,i}N_{m+1}$$ with
  $s_{h,i}\in S_h^\star$.
  We set $s_{h,F}=s_{h,1}\cdot s_{h,2}\cdots s_{h,p}$, and
  observe, because of $\varepsilon\in \varphi(w)$ for all $w\in S_h$
  and $S_h$,
  $$\varphi(s_{h,i})= \varepsilon\cdot\varphi(s_{h,i})\cdot\varepsilon\subseteq
  \varphi(s_{h,{1}})\cdots\varphi(s_{h,{i-1}})\cdot
  \varphi(s_{h,i})\cdot\varphi(s_{h,{i+1}})\cdots\varphi(s_{h,p})=
  \varphi(s_{h,F})\;.$$
  Thus we choose the summary word $w=N_1s_{1,F}N_2\dots
  N_ms_{m, F}N_{m+1}$ and obtain
  $\varphi(f_i)=\varphi(N_1s_{1,i}N_2\dots
  N_ms_{m,i}N_{m+1})\subseteq \varphi(w)$, and hence
  $\varphi(F)\subseteq\varphi(w)$ .
  \qed
\end{proof}

\begin{proof}[\textbf{of Proposition~\ref{prop:unfold-rewriting}}]
  We have $S_h^\star=E^\star(\bar E E^\star)^\star=E^\star\; \cup \;
  E^\star\bar E E^\star(\bar E E^\star)^\star=E^\star\; \cup \;
  E^\star\bar E S_h^\star$
  and find the desired result by substituting $\bar E= \bigcup_{p\in
    \criticalpos(S_h)} \bar E_p$, as discussed before
  Proposition~\ref{prop:unfold-rewriting}.

  \begin{inparaenum}[\bfseries(1)]
  \item \emph{Union freeness:} We construct the regular expression for
    $E^\star$ by dropping all Kleene-stared subexpressions in
    $S_h^\star$ which contain a symbol $\delta$ with
    $\varepsilon\not\in\varphi(\delta)$ (possibly producing the empty
    language), preserving union freeness.
    The construction of $\bar E_p$ only unrolls Kleene star expressions,
    also preserving the union freeness  from $S_h$. 
  \item \emph{1-word summaries for $E^\star$:} For all $w\in E$, we
    have $\varphi(w)=\varepsilon$, since all symbols $\delta$ in $E$
    have $\varepsilon\in\varphi(\delta)$.
  \item \emph{Increasing minimal length in $E^\star\bar E_p
      S_h^\star$:} Since $S_h^\star$ is a subexpression of
    $E^\star\bar E_p S_h^\star$ the minimal length can only increase,
    and since $\bar E_p$ instantiates an expression with a symbol
    $\delta$ and $\varepsilon\not\in\varphi(\delta)$, it actually
    increases. \qed
  \end{inparaenum}
\end{proof}

\subsubsection{Proof of Theorem~\ref{thm:membership-complexity}.}
The proof is based on the complexity of the maximum rewriting
from~\cite{rewriting} and the complexity of \unfold, shown first via
Propositions~\ref{prop:unfold-step-size}
and~\ref{prop:unfold-complexity}, before proving the overall
complexity of \setenumeration in Lemma~\ref{lem:complexity-enumerate}.
This Lemma, together with
Lemma~\ref{lem:pspace-completiy-of-prop1-check}, leads to the desired
theorem.

Recall the definition of \unfoldstep before
Proposition~\ref{prop:unfold-rewriting} for
$L=\alpha_1\beta_1^\star\alpha_2 \dots
\alpha_n\beta_n^\star\alpha_{n+1}$ and $p=\left<p_H\mid p_T\right>$
with $\unfoldstep(L, p) = \alpha_1\dots \alpha_{p_H} \beta_{p_H}^\star
\unfoldstep(\beta_{p_H},p_T) \beta_{p_H}^\star
\alpha_{p_H+1}\dots\alpha_{n+1}$.
We denote with $||L||$ the \emph{length of the regular expression}
representing $L$.

\begin{proposition}[An Upper Bound for $||\unfoldstep(L,p)||$]
  \label{prop:unfold-step-size}
  Let $K$ be the maximum length of a Kleene star subexpression in
  $L$. 
  Then $||\unfoldstep(L,p)||=\order(K||L||)$ holds.
\end{proposition}

\begin{proof}
  \unfoldstep duplicates $\beta_{p_H}$ of $L$ and continues
  recursively on a third copy of $\beta_{p_H}$.
  Since \unfoldstep does not introduce new Kleene star subexpressions
  but only duplicates some, all Kleene star expressions occurring
  during the \emph{entire} recursion are at most of length $K$.
  Hence, each recursive step of $\unfoldstep$ adds at most $2K$ to the
  entire expression, and because the Kleene star nesting depth of at
  most $L$, we obtain $||\unfoldstep(L,p)||= \order(K||L||)$.\qed
\end{proof}

\begin{proposition}[$\unfold(L,\varphi,B)$ runs in $\DSPACE(B^2
  ||L||^4+||\varphi||)$]
  \label{prop:unfold-complexity}
\end{proposition}

\begin{proof}
  In this proof, we denote with \Linit the language given in the first
  call to \unfold, while $L$ denotes the language given to current
  call of \unfold. 
  We show the claim in three steps: 
  \begin{inparaenum}[\bfseries(1)]
  \item $||L||=\order(d||\Linit||^2)$ holds at any point during the
    recursion, given $d$ is the number of recursive calls going through
    Line~\ref{line:unfold-rec-iterate}.
    First, recursive calls through Line~\ref{line:unfold-rec-0} cannot
    increase the size of the expression, i.e., $||L_0||\le ||L||$,
    since we obtain $L_0$ by removing from $S_h^\star$ all
    subexpressions directly containing a symbol $\delta$ with
    $\varepsilon\in\varphi(\delta)$ (and not only via another Kleene
    star expression).
    Thus, only recursive calls going through
    Line~\ref{line:unfold-rec-iterate} possibly increase the size of
    the expression.
    Now, in such a call, we unroll a subexpression $S_h$ with
    $S_h^\star=E^\star\bar E_pS_h^\star$ and $\bar
    E=\unfoldstep(S_h,p)$ for some integer sequence $p$.
    From Proposition~\ref{prop:unfold-step-size}, we have
    $||\unfoldstep(S_h,p)||=\order(K||S_h||)$. 
    Since \unfoldstep and \unfold only duplicate already existing
    Kleene star subexpressions, we have both $||S_h||\le ||\Linit||$
    and $K\le ||\Linit||$, and hence
    $||\unfoldstep(S_h,p)||=\order(||\Linit||^2)$.
    Together with $||E||\le ||S_h|$ and $||S_h||\le ||\Linit||$, this
    leads to $||E^\star\bar E_pS_h^\star||=\order(||\Linit||^2)$.
    $d$ recursive calls through Line~\ref{line:unfold-rec-iterate}
    substitute $d$ subexpressions $S_h$ with $E^\star\bar
    E_pS_h^\star$ to unfold $\Linit$ into $L$, each time adding
    $\order(||\Linit||^2)$ to the size of the expression representing
    $L$.
    Hence $||L||=\order(d||\Linit||^2)$.

  \item $||L||=\order(B||\Linit||^2)$ holds for all recursive calls to
    \unfold while computing $\unfold(\Linit,\varphi,B)$.
    \unfold makes at most $B$ recursive steps through
    Line~\ref{line:unfold-rec-iterate}, since
    $\minlength(\varphi(L_p))>\minlength(\varphi(L))$ holds (this is
    true, since $\bar E_p$ in $L_p$ instantiates some $\delta$ with
    $\varepsilon\not\in\varphi(\delta)$).
    Then the claim follows setting $d=B$.

  \item The total recursion depth of \unfold is at most
    $\order(B||\Linit||^2)$.
    In the previous claim, we saw that there are at most $B$ recursive
    calls through Line~\ref{line:unfold-rec-iterate}.
    It remains to give an upper bound for the calls through
    Line~\ref{line:unfold-rec-0}:
    In each such call, at least one Kleene star subexpression in $L$
    is removed in substituting $E$ for $S_h$.
    At any point there are at most $||L||=\order(B||\Linit||^2)$
    expressions in $L$, hence we get a maximum recursion depth of
    $\order(B||\Linit||^2)$.

  \item The space required to compute $\unfold(\Linit,\varphi,B)$ is
    bounded by the depth of the recursion times the stack frame size,
    which is dominated by $||L||$, plus $||\varphi||$.
    This gives $\order\left( (B||\Linit||^2)^2+||\varphi||\right)=
    \order(B^2||\Linit||^4+||\varphi||)$ as desired.\qed
  \end{inparaenum}
\end{proof}

\begin{lemma}[$\setenumeration(R,K,\varphi)$ runs in
  $\DSPACE\left(||K||^4 2^{2^{(||R||+||\varphi||)^k}}\right)$]
  \label{lem:complexity-enumerate}
\end{lemma}

\begin{proof}
  The construction of $M=M_\varphi(R)\cap K$ yields an expression in
  the size $||K||2^{2^{(||R||+||\varphi||)^l}}$ for some constant
  $l$~\cite{rewriting}.
  The union free decomposition yields possibly exponentially many
  union free languages, however, each of them has linear size, using
  the rewriting rules, $(A+B)C=AC+BC$, $A(B+C)=AB+AC$,
  $(A+B)(C+D)=AC+AD+BC+BD$, and $(A+B)^\star=(A^\star B^\star)^\star$.
  In practical implementations, however, one might prefer to generate
  less but larger individual expressions, employing e.g.~\cite{decomposition}. 
  With Proposition~\ref{prop:unfold-complexity}, we obtain the overall
  space complexity of \setenumeration with $\DSPACE(B^2
  ||L||^4+||\varphi||)$ for $B=\minlength(R)\le ||R||$ and
  $||L||=||K||2^{2^{(||R||+||\varphi||)^k}}$.
  %
  %
  This leads to the desired result with $\DSPACE\left(||K||^4 2^{\cdot
      2^{(||R||+||\varphi||)^k}}\right)$ for some other constant
  $k$.\qed
\end{proof}

\begin{proof}[\textbf{of Theorem~\ref{thm:membership-complexity}}]
  The enumeration runs $\DSPACE\left(||K||^4
    2^{2^{(||R||+||\varphi||)^k}}\right)$, producing expressions for
  $\basiccheck$ at most of the same size
  (Lemma~\ref{lem:complexity-enumerate}).
  Since $\basiccheck$ is in $\PSPACE$
  (Lemma~\ref{lem:pspace-completiy-of-prop1-check}), we obtain the
  overall complexity $\DSPACE\left(||K||^r
    2^{2^{(||R||+||\varphi||)^s}}\right)\subseteq\tEXPSPACE$ for some
  constants $r$ and $s$. \qed
\end{proof}

\end{document}